\documentclass{amsart}

\usepackage[T1]{fontenc}
\usepackage[utf8]{inputenc}
\usepackage{amssymb}
\usepackage{lmodern}
\usepackage{overpic}
\usepackage[outline]{contour}
\usepackage{mathtools}
\usepackage{microtype}
\usepackage[sort&compress,numbers,square]{natbib}
\usepackage{todonotes}
\usepackage{enumitem}
\setlist[enumerate,1]{label=\arabic*.,ref=\arabic*.}

\usepackage{algorithm}
\usepackage{algpseudocode}

\algrenewcommand\algorithmicwhile{\textbf{While}}
\algrenewcommand\algorithmicfor{\textbf{For}}
\algrenewcommand\algorithmicdo{\textbf{Do}}
\algrenewcommand\algorithmicif{\textbf{If}}
\algrenewcommand\algorithmicthen{\textbf{Then}}
\algrenewcommand\algorithmicelse{\textbf{Else}}
\algrenewcommand\algorithmicend{\textbf{End}}
\algrenewcommand\algorithmicreturn{\textbf{Return}}

\usepackage{hyperref}
\hypersetup{
  pdfauthor={Zijia Li, Josef Schicho, Hans-Peter Schröcker},
  pdftitle={Kempe's Universality Theorem for Rational Space Curves},
  pdfkeywords={},
  hidelinks,
}

\newtheorem{theorem}{Theorem}
\newtheorem{lemma}{Lemma}
\newtheorem{corollary}{Corollary}
\newtheorem{proposition}{Proposition}
\theoremstyle{definition}
\newtheorem{definition}{Definition}
\theoremstyle{remark}
\newtheorem{remark}{Remark}
\newtheorem{example}{Example}

\title[Kempe's Universality Theorem for Rational Space Curves]{Kempe's Universality Theorem\\for Rational Space Curves}
\date{\today}

\author{Zijia Li}
\address[Zijia Li]{Joanneum Research, Institute for Robotics and Mechatronics, Lakeside B08a, 9020 Klagenfurt, Austria, 
  Phone +43 316 876 2016}
\urladdr{https://www.joanneum.at/nc/en/get-to-know-us/employees/detail/staff/Li-Zijia/}
\email{zijia.li@joanneum.at}

\author{Josef Schicho}
\address[Josef Schicho]{Research Institute for Symbolic Computation,
  Johannes Kepler University Linz, Schloss Hagenberg, 4232 Hagenberg, Austria}
\urladdr{http://www.risc.jku.at/people/jschicho/}
\email{josef.schicho@risc.jku.at} 

\author{Hans-Peter Schröcker}
\address[Hans-Peter Schröcker]{Unit Geometry and CAD, University of Innsbruck, Technikerstr.~13, 6020 Innsbruck, Austria}
\urladdr{http://geometrie.uibk.ac.at/schroecker/}
\email{hans-peter.schroecker@uibk.ac.at}


\newcommand{\C}{\mathbb{C}}
\newcommand{\D}{\mathbb{D}}
\renewcommand{\H}{\mathbb{H}}
\renewcommand{\P}{\mathbb{P}}
\newcommand{\R}{\mathbb{R}}
\newcommand{\Q}{\mathbb{Q}}
\renewcommand{\DH}{\D\H}
\newcommand{\MDH}{\mathbb{S}}
\newcommand{\qi}{\mathbf{i}}
\newcommand{\qj}{\mathbf{j}}
\newcommand{\qk}{\mathbf{k}}
\newcommand{\eps}{\varepsilon}
\newcommand{\cj}[1]{\overline{#1}}
\newcommand{\Norm}[1]{N(#1)}
\newcommand{\SE}[1][3]{\mathrm{SE}(#1)}
\newcommand{\SQ}{S}
\newcommand{\axis}[1]{\ell(#1)}
\newcommand{\lplus}[2]{\ensuremath{\text{concat}(#1, #2)}}
\newcommand{\lmin}[2]{\ensuremath{\text{remove}(#1, #2)}}

\newcommand{\ffont}[1]{#1}
\DeclareMathOperator{\mrpf}{\ffont{mrpf}}
\DeclareMathOperator{\lcoeff}{\ffont{lcoeff}}
\DeclareMathOperator{\rrem}{\ffont{rrem}}
\DeclareMathOperator{\rQR}{\ffont{rQR}}
\DeclareMathOperator{\quo}{\ffont{quo}}
\DeclareMathOperator{\lquo}{\ffont{lquo}}
\DeclareMathOperator{\rquo}{\ffont{rquo}}
\DeclareMathOperator{\lgcd}{\ffont{lgcd}}
\DeclareMathOperator{\rgcd}{\ffont{rgcd}}
\DeclareMathOperator{\minmot}{\ffont{minmot}}
\DeclareMathOperator{\tfactor}{\ffont{tfactor}}
\DeclareMathOperator{\gfactor}{\ffont{gfactor}}
\DeclareMathOperator{\czero}{\ffont{czero}}

\DeclareMathOperator{\bflip}{\ffont{bflip}}
\DeclareMathOperator{\minpol}{\ffont{mp}}

\begin{document}

\begin{abstract}
  We prove that every bounded rational space curve of degree $d$ and circularity
  $c$ can be drawn by a linkage with $ \frac{9}{2} d-6c+1$ revolute joints. Our
  proof is based on two ingredients. The first one is the factorization theory
  of motion polynomials. The second one is the construction of a motion
  polynomial of minimum degree with given orbit. Our proof also gives the
  explicit construction of the linkage.
\end{abstract}

\maketitle

\section{Introduction}
\label{sec:introduction}

Kempe's Universality Theorem \cite{kempe76} is one of the great
theorems of theoretical mechanism science (``beautiful''
\cite{blaschke56,demaine07}, ``surprising'' \cite{abbott08,demaine07},
``incredible theoretical significance'' \cite{saxena11}, ``shocking''
\cite{gao01}). It states that any bounded portion of a planar
algebraic curve can be traced out by one joint of a planar linkage
with revolute joints. Discovered only shortly after the invention of
the first straight line linkages, Kempe's theorem must have been a
true surprise to his contemporary kinematicians. Throughout the 20th
century, it was considered a milestone result.

Kempe's constructive proof can be used to actually compute a linkage
that draws a planar algebraic curve. However, it was clear from the
beginning that this construction is of no practical relevance.  It
requires an excessive number of links and joints, even for curves of
low degree. Nowadays, an asymptotic bound of $O(d^n)$ for the number
of links necessary to draw an algebraic curve of degree $d$ in an
ambient space of dimension $n$ is known \cite{abbott08}. Nonetheless,
drawing an ellipse with a Kempe linkage already requires hundreds of
links \cite{kobel08}. A wealth of more practical examples for algebraic curve
generation can be found in the monograph
\cite{artobolevskii64}. However, the constructions there are rather
specific to certain classes of curves and, more importantly, use
mechanical constraints different from rigid links and revolute joints.

Kempe's Theorem talks about algebraic curves and it is natural to ask
for simplifications in case of rational curves. This was done recently
in \cite{gallet15} where the authors constructed scissor-like linkages
to draw rational planar curves. Their construction is based on the
factorization of certain polynomials over a non-commutative ring that
describe the motion of one of the links. The upper bounds on the
number of links and joints for curves of degree $d$ reduce
dramatically to $3d+2$ and $\frac{9}{2}d+1$, respectively. In this article, we
extend the ideas of \cite{gallet15} to rational space curves. Our aim
is to provide a construction that works for all rational space curves
and, at the same time, to reduce the number of links and joints as far
as possible. For this purpose, we introduce several new ideas that
also improve the planar case.

We use rational motions of minimal degree in the dual quaternion model
of rigid body displacements for the link that draws the given space
curve. This acknowledges the importance of circularity. If the
rational space curve is entirely circular, the motion degree is
particularly low \cite{li15b,li15a} and factorization of the motion
polynomial is straightforward, without the need for prior degree
elevation. In the non-circular case, a degree elevation is necessary
but in contrast to \cite{gallet15,li15b}, we only preserve one
relevant trajectory, not the complete rational motion. This allows to
keep the degree lower and saves links and joints.  If the curve is of
degree $d$ and circularity $c$, the bounds for links and joints are
$3d-4c+2$ and $\frac{9}{2}d-6c+1$, respectively.

Another advantage of our spatial approach concerns a certain defect in
Kempe's original construction that later was even considered a flaw
\cite[Section~3.2]{demaine07}. Kempe used parallelogram and
anti-parallelogram linkages as basic building blocks of his
linkages. It is well known that the configuration curves of these
linkages consists of two irreducible algebraic components that
intersect at flat folded positions. That is, the linkage may switch
between parallelogram and anti-parallelogram mode, thus entering
unwanted components of the configuration curve. The effect is that
Kempe's linkages draw more than the originally intended curve. At the
cost of introducing additional links and joints, this defect can be
overcome by the ``bracing constructions'' of
\cite{kapovich02,demaine07}. The approach of \cite{gallet15} uses
anti-parallelograms and is therefore subject to the same defect
and its resolution.

The basic building blocks of our approach are Bennett linkages whose
configuration curve has only one irreducible component. Thus, no
additional links and joints are needed to prevent the linkage from
switching modes. We also capture this in the phrase ``the
configuration curve is free of spurious components''.

This article uses several results from other, recently published,
papers. Whenever we use such a result we give a concise summary and
references. Moreover, we provide algorithmic descriptions so that a
reader of this paper will be able to construct linkages for drawing an
arbitrary rational space curve. We continue this article with a
formulation of the main theorem and an overview of its proof in
\autoref{sec:overview}. In \autoref{sec:preliminaries} we provide a
concise introduction to dual quaternions and motion polynomials. The
proof of our main theorem is done in \autoref{sec:construction}. It is
subdivided into several steps: Construction of a minimal motion to the
given rational curve, factorization of this motion, and subsequent
linkage construction by means of ``Bennett flips''. In
\autoref{sec:examples} we present several examples to illustrate
important points of our construction. In the concluding
\autoref{sec:discussion} we discuss our result and ideas. We mentioned
implementation issues and outline possible extension and application.

\section{Main theorem and overview of proof }
\label{sec:overview}

Our main result in this paper is a statement about linkages and
bounded rational curves in three-space. A rational curve is a curve
admitting a parametric equation of the shape
$X = x_0^{-1}(x_1,x_2,x_3)$ with polynomials
$x_0,x_1,x_2,x_3 \in \R[t]$ and $x_0 \neq 0$. It is no restriction to
assume that this parametric equation is \emph{reduced,} that is,
$\gcd(x_0,x_1,x_2,x_3) = 1$ because we may always divide by a common
factor. The \emph{degree} of the rational curve is the maximum of the
polynomials $x_0$, $x_1$, $x_2$, and $x_3$ in reduced form. The
\emph{circularity} of a reduced rational parametric equation is
$c \coloneqq \frac{1}{2}\deg\gcd(x_0,x_1^2+x_2^2+x_3^2)$. It counts
the number of intersection points with the absolute conic of Euclidean
geometry, is a positive integer and invariant with respect to rational
re-parameterizations and similarity transformations. Finally, the
rational curve is called \emph{bounded,} if $x_0$ has no real zeros
and $\deg(x_0)\ge\deg(x_1),\deg(x_2),\deg(x_3)$. Note that in this
sense bounded segments of unbounded rational curves are bounded
rational curves.

A linkage in our context consists of a set of lines in space, called
the \emph{joints}, and links (rigid bodies which could have different
shapes) that connect two or more axes. Whenever two links are
connected by a joint, their relative position is constrained to a
rotation about their common joint. More precisely, the relative
position is determined by the rotation angle about this joint with
respect to a given reference configuration. The set of all tuples of
possible rotation angles is called the linkage's \emph{configuration
  space.} If it is of dimension one, we call it a \emph{configuration
  curve} and say the linkage has \emph{one degree of freedom.} If this
is the case, we designate one link as fixed and another as moving. We
view the relative displacements of the moving link with respect to the
fixed link as a curve in $\SE$. The \emph{orbit} or \emph{trajectory}
of a point attached to the moving link is the locus of all positions
in space when the point is subject to all possible displacements in
this curve. A linkage is called \emph{spherical,} if all axes are
concurrent and \emph{planar} if they are parallel. The trajectories of
spherical linkages are spherical curves and the trajectories of planar
linkage are curves in parallel planes. For a more formal definition of
(at least planar) linkages we refer to \cite{gallet15}.

Let us illustrate some of these concept at hand of a Bennett linkage
\cite{bennett03,bennett14,baker98,perez04} which will play a crucial
role in our linkage construction in \autoref{sec:scissor-linkage}. A
Bennett linkage is a spatial four-bar linkage with one degree of
freedom. Its four axes $\ell_1$, $\ell_2$, $\ell_3$, $\ell_4$ are the
perpendiculars to two incoming edges in the vertices of a spatial
parallelogram, that is, a spatial quadrilateral with equal opposite
edge lengths. \autoref{fig:bennett} displays an abstract
representation and a 3D model of a Bennett linkage. The joint axes
$\ell_1$ and $\ell_4$ are attached to the fixed link while $\ell_2$
and $\ell_3$ belong to the moving link. Two things about Bennett
linkages are important to us:
\begin{itemize}
\item Bennett linkage constitute the only type of spatial four-bar
  linkages with one degree of freedom and
\item their configuration curve consists of only one component.
\end{itemize}

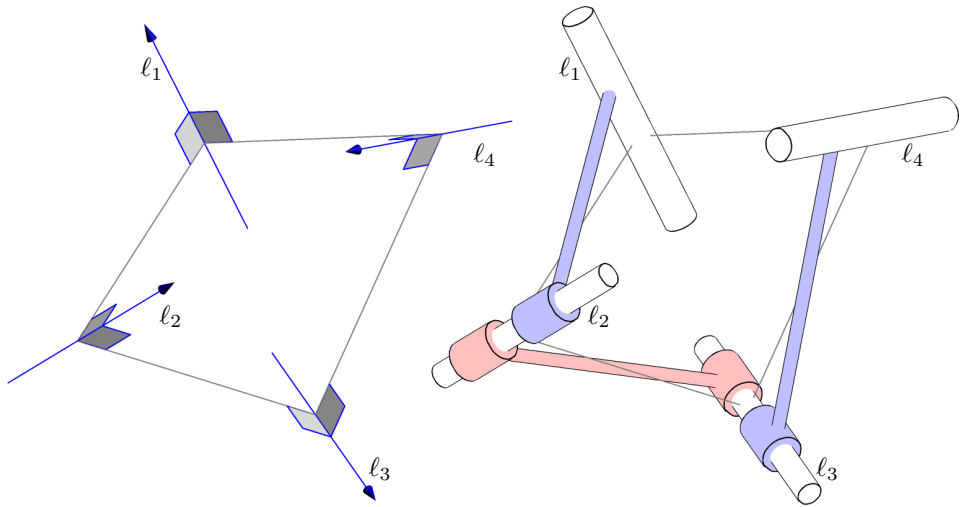
\begin{figure}
  \centering
  \begin{overpic}[width=\textwidth]{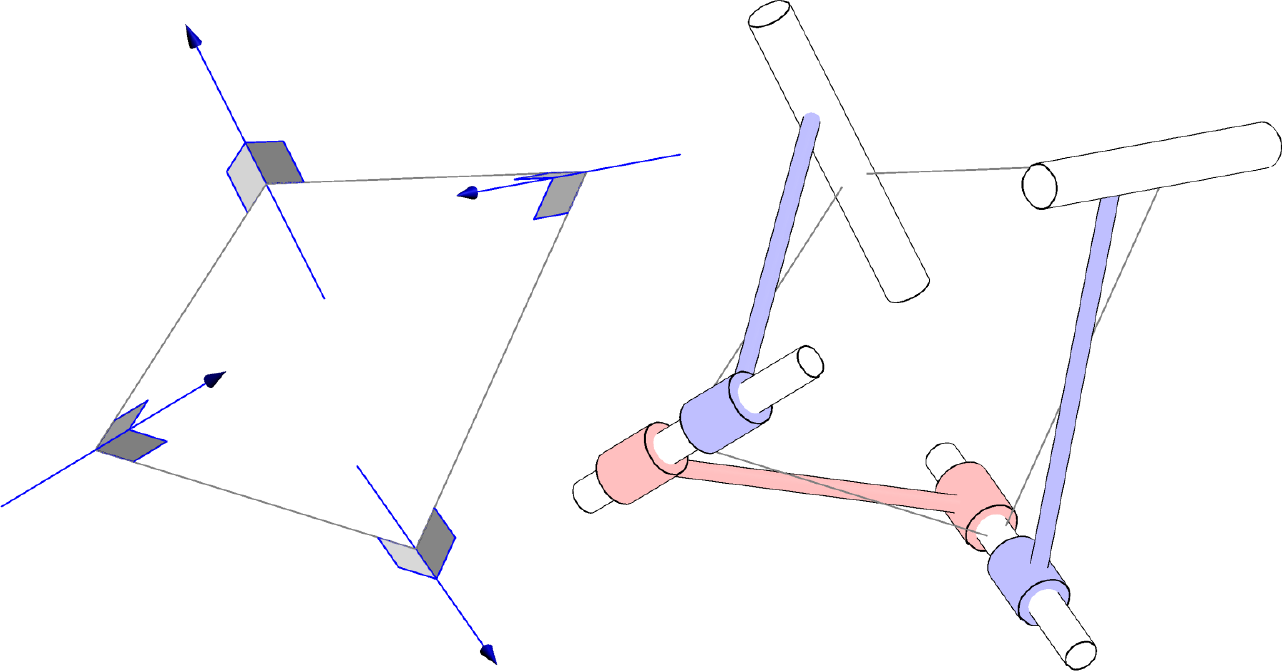}
    \put(14,45){$\ell_1$}
    \put(16,19){$\ell_2$}
    \put(38,3){$\ell_3$}
    \put(49,36){$\ell_4$}
    \put(58,45){$\ell_1$}
    \put(61,19){$\ell_2$}
    \put(85,3){$\ell_3$}
    \put(94,36){$\ell_4$}
  \end{overpic}
  \caption{Bennett linkage}
  \label{fig:bennett}
\end{figure}

Here is our main result:

\begin{theorem}
  \label{th:1}
  For every bounded rational curve of degree $d$ and circularity $c$
  in three space, there exists a spatial revolute linkage with one
  degree of freedom with at most $3d-4c+2$ links and
  $\frac{9}{2}d-6c+1$ joints such that the trajectory of one point
  attached to the moving link is precisely the given rational curve.
\end{theorem}

Before we embark on a proof of this theorem (which will consume large
parts of this paper), we would like to make a few remarks:
\begin{itemize}
\item The theorem talks about \emph{bounded} rational curves. This is
  necessary because all trajectories of a linkage with only revolute
  joints are bounded. It is, however, admissible that the rational
  curve parameterizes a bounded portion of an unbounded curve, for
  example a line segment.
\item Boundedness also implies that $d$ is even and $c$ is an integer
  so that our bounds on the number of links and joints are integers as
  well.
\item The linkage of \autoref{th:1} is not unique. In fact,
  uncountably many linkages exist. Even planar or spherical curves may
  be drawn by spatial linkages.
\item However, our approach is capable of producing planar linkages
  for planar curves (\autoref{cor:2}) and \autoref{th:1} retains the
  bounds of \cite{gallet15} for the numbers of links and joints of
  bounded planar rational curves.
\item Our approach can also produce spherical linkages for spherical
  curves. For them, the number of links and joints reduces to $d+2$
  and $\frac{3}{2}d+1$, respectively (\autoref{cor:spherical}).
\item The linkages can be constructed in such a way, that their
  configuration curve is free of spurious components. This is however
  not the case if the linkage is required to be planar or spherical.
\end{itemize}

Our proof of \autoref{th:1} is constructive and can be translated into
algorithms, mostly based on polynomial algebra over the dual
quaternions. Our presentation pays attention to these algorithmic approach
and gives ample information for actual implementation. The proof
consists of several steps. We begin by constructing a rational motion
such that one point has the given rational curve as trajectory. In
order to keep small the number of links and joints, it is advantageous
to require a minimal motion degree in the dual quaternion model of
rigid body displacements. In \cite{li15a}, we proved that these
properties determine a unique rational motion which may be computed in
rather straightforward manner by \autoref{alg:minmot}.

This rational motion is parameterized by a certain polynomial with
dual quaternion coefficients which we call a ``motion
polynomial''. General motion polynomials can be written, in several
ways, as products of linear motion polynomials \cite{hegedus13}. These
factorizations correspond to the decomposition of the rational motion
into products of rotations. The axes are determined by the linear
motion polynomials and the rotation angles are linked by a common
parameter. On these open chains of revolute joints we base our linkage
construction. Unfortunately, there are motion polynomials that do not
allow factorizations. Even worse, the motion polynomials of minimal
degree with prescribed trajectory typically fall into this
category. Therefore, we have to artificially raise their degree in
such a way that factorizations exist. This can be done by multiplying
them with real polynomials as in \citep{gallet15,li15b}. This does not
change the underlying motion and may be advantageous in certain
applied situations. However, the bounds in \autoref{th:1} are only
obtained by a refinement of this procedure. We right-multiplying the
motion polynomial with a certain quaternion polynomial. This
changes the motion but not the trajectory in question. A curious
side-effect of this approach is that it may turn a planar or spherical
motion into a spatial motion.

Having constructed a factorizable motion polynomial to the prescribed
trajectory, we have to construct a linkage. In general, it is possible
to combine the open chains obtained from different factorizations to
form a linkage with one degree of freedom \cite{hegedus13, li13a, li14ck, li14ark}. 
 We do, however, not pursue
this approach because it seems difficult to prove that the resulting
linkage has \emph{always} (not just in general) only a single degree
of freedom. Moreover, spurious motion components do exist
 \cite{li13a, li14ck}. Instead, we adapt the scissor-linkage construction of
\cite{gallet15} to the spatial case, replacing the anti-parallelograms
of planar linkages by spatial Bennett linkages. This automatically
guarantees that the configuration space is of dimension one and has no
spurious components.

It should also be mentioned that all algorithms presented here require
exact (zero error) computation, which means symbolic methods. These are
not possible for $\R$, but only for suitable real closed subfields, in
particular for the set of real algebraic numbers. Real closure is necessary
because we will need to factor a univariate polynomial into its
irreducible linear and quadratic factors. In the examples, we
try to remain in subfields for which arithmetic does not get too complicated,
such as $\Q$ or real quadratic extensions.

\section{Dual quaternions and kinematics}
\label{sec:preliminaries}

This section provides an introduction to dual quaternions and their
relation to space kinematics. In particular, we introduce a
homomorphism between a certain subgroup of dual quaternions into $\SE$
and the important concept of motion polynomials.

Denote by $\H$ the non-commutative ring of quaternions. An element
$h \in \H$ can be written as $h = h_0 + h_1\qi + h_2\qj + h_3\qk$.
The quaternion units $\qi$, $\qj$, and $\qk$ satisfy the
multiplication rules
\begin{equation*}
  \qi^2 = \qj^2 = \qk^2 = \qi\qj\qk = -1.
\end{equation*}
By $\D$ we denote the ring $\R[\eps]/\langle \eps^2 \rangle$. Its
elements are called \emph{dual numbers.} The scalar extension
$\DH \coloneqq \D \otimes_\R \H$ of $\H$ by $\D$ gives the ring of
dual quaternions.

A dual quaternion $h$ may be written as $h = p + \eps q$ with
quaternions $p$ and $q$, the \emph{primal} and \emph{dual part} of
$h$, respectively. The conjugate dual quaternion is
$\cj{h} = \cj{p} + \eps\cj{q}$ and quaternions in $\H$ are conjugated
by multiplying the coefficients of $\qi$, $\qj$, and $\qk$ with
$-1$. The dual quaternions with non-zero primal part are
invertible. The inverse of $h$ is $h^{-1} = \Norm{h}^{-1}\cj{h}$. Here,
$\Norm{h} \coloneqq h\cj{h} = p\cj{p} + \eps(p\cj{q} + q\cj{p})$
denotes the \emph{dual quaternion norm.} It is a dual number and,
provided $a \neq 0$, the inverse of number $a + \eps b \in \D$ is
$a^{-1} - \eps ba^{-2}$.

Denote by $\MDH$ the multiplicative subgroup of dual quaternions with
real, nonzero norm. It acts on $\R^3 = \langle \qi, \qj, \qk \rangle$
according to
\begin{equation}
  \label{eq:1}
  z \mapsto \frac{pz\cj{p}+p\cj{q}-q\cj{p}}{\Norm{p}}.
\end{equation}
This equation defines a homomorphism from $\MDH$ to $\SE$. It is
surjective and the kernel is the real multiplicative group $\R^\star$.
Hence, there exists an isomorphism between $\MDH/\R^\star$ and $\SE$.
This is actually Study's well-known kinematic mapping (or its
inverse). Factorizing by $\R^\star$ turns $\DH$ into real projective
space $\P^7$ and $\MDH$ becomes the Study quadric $\SQ \subset \P^7$
minus the exceptional three-space of classes of dual quaternions with
vanishing primal part. More details can be found for example in
\cite{husty09:_algebraic_geometry_kinematics}.

Now we make this group homomorphism parametric. Denote by $\DH[t]$ the
skew ring of polynomials over $\DH$ with indeterminate $t$. We define
multiplication in this ring by the convention that $t$ commutes with
all coefficients. This is a natural convention because $t$ will later
act as a real motion parameter and $\R$ is in the center of
$\DH$. Some notions that have already been defined for $\DH$ can be
transferred to $\DH[t]$. For $C \in \DH[t]$ the \emph{conjugate
  polynomial} $\cj{C}$ is obtained by conjugating all coefficients of
$C$. If $C = P + \eps Q$ with $P, Q \in \H[t]$, then $P$ and $Q$ are
called \emph{primal} and \emph{dual part,} respectively. The norm
polynomial is
$\Norm{C} \coloneqq C\cj{C} = P\cj{P} + \eps(P\cj{Q} + Q\cj{P})$. Its
coefficients are dual numbers. If $C = \sum_{i=0}^n c_it^i$, the value
of $C$ at $h \in \DH$ is defined as
$C(h) \coloneqq \sum_{i=0}^n c_ih^i$. With these definitions,
evaluation of polynomials at a fixed value $h \in \DH$ is not a ring
homomorphism. Still, the dual quaternion zeros of polynomials over
$\DH$ have a meaning in our algorithms.
 
\begin{definition}
  The polynomial $C \in \DH[t]$ is called a \emph{motion polynomial,}
  if $C\cj{C} \in \R[t] \setminus \{0\}$ and if its leading
  coefficient $\lcoeff(C)$ is invertible.
\end{definition}

Motion polynomials are a central concept in this article. Their
introduction is motivated by the fact that for every $t_0 \in \R$, the
value $C(t_0)$ is an element of $\MDH$ so that $C$ acts on
$\R^3$. Varying $t_0$, we get a one-parametric set of rigid body
displacements, that is, a motion. By virtue of \eqref{eq:1}, the orbit
of any point is (part of) a rational curve, whence the motion itself
is called \emph{rational.} It is possible to extend the parameter range
from $\R$ to $\R \cup \{\infty\}$: with the definition 
$C(\infty) \coloneqq \lcoeff(C)$, the parametrization of the orbit
is continuous except in the points $t_1$ such that $C(t_1)$ has norm zero.
One could also represent the map as a regular map from the real projective
line to real projective space $\P^7$, but this would require a second
homogeneous variable for the parameter $t$, which complicates the
algebraic theory; hence, we prefer to use an affine parameter space
and a projective image.

Of particular importance to us are linear motion polynomials. The
linear polynomial $C = t - h$ is a motion polynomial if
$C\cj{C} = t^2 - (h + \cj{h})t + h\cj{h}$ is real. This is the case if
both $h + \cj{h}$ and $h\cj{h}$ are real. It is well-known that the
motion parameterized by $C$ is either a rotation about a fixed axis
or, if the primal part of $h$ is real, a translation in fixed
direction. In either case, the parameter value $t = \infty$
corresponds to the identity transformation, that is, zero rotation
angle or translation distance.

In this article, we often assume that a given motion polynomial has no
nontrivial real factor. We call those motion polynomials
\emph{reduced.} From a kinematic viewpoint, this is no restriction as
$C$ and $CR$ with $R \in \R[t] \setminus \{0\}$ parameterize the same
motion. Note, however, that multiplication with a real polynomial is a
useful technique to ensure existence of a factorization (see
\cite{gallet15,li15b} and \autoref{sec:factorization}).

\section{Proof of main theorem, linkage construction}
\label{sec:construction}

In this section we prove \autoref{th:1}. For each step in our
constructive proof we provide a theoretical justification, often with
references to existing literature, and an algorithmic
description. Given is a bounded rational curve by a reduced rational
parametric equation $X = x_0^{-1}(x_1,x_2,x_3)$. We also encode it as
polynomial $x = x_0 + x_1\qi +x_2\qj + x_3\qk \in \H[t]$.  The
parametric equation $X$ and the polynomial $x$ are of the same
degree~$d$.

\subsection{Motion of minimal degree to given trajectory}
\label{sec:minimal-degree}

The first step is the construction of a rational motion such that the
trajectory of one point equals the parametric curve $X$. An obvious
and simple choice is the translation along $X$, given by the motion
polynomial $C = x_0 - \frac{1}{2}\eps (x_1\qi +x_2\qj + x_3\qk)$. However, in order to keep low
the number of links and joints, we try to find a motion polynomial $C$
of minimal degree with trajectory $X$. The main result in this context
states uniqueness of this motion and characterizes curves for which
the trivial translation along the curve is not optimal \cite{li15a}.
In order to fully appreciate it, we need a definition:

\begin{definition}
  \label{def:trajectory-degree}
  The \emph{trajectory degree} of a reduced rational motion
  $C = P + \eps Q \in \DH[t]$ is the maximal degree of its
  trajectories. The \emph{quaternion degree} is the degree of $C$ as
  polynomial in $\DH[t]$. The \emph{spherical degree defect} of $C$ is
  the degree of the real polynomial factor of maximal degree of the
  primal part $P$.
\end{definition}

Some authors refer to the trajectory degree of a reduced rational
motion as just the ``degree''. However, we have to distinguish between
this trajectory degree and the degree of $C$ as polynomial in
$\DH[t]$. The latter was called the motion's ``quaternion degree'' in
\cite{juettler93} and we follow this convention. The spherical degree
defect accounts for a difference in the respective trajectory degrees
of $C$ and its spherical motion component $P$. If the motion
polynomial is monic (or at least its leading coefficient is
invertible), the spherical degree defect can be computed as degree of
$\mrpf P \coloneqq \gcd(P,\cj{P})$ where $\gcd$ denotes the monic real
polynomial factor of maximal degree.

\begin{theorem}[\cite{li15a}]
  \label{th:2}
  The rational motion of minimal quaternion degree with a prescribed
  rational trajectory is \emph{unique.} If the trajectory is of degree
  $d$ and circularity $c$, this minimal motion is of degree $d-c$ and
  has a spherical degree defect of $s = d - 2c$.
\end{theorem}

Lets call the motion of \autoref{th:2} the trajectory's \emph{minimal
  motion.} \autoref{th:2} tells us that in general (if the trajectory
is of circularity zero), the translation along the trajectory is
minimal. But for curves with positive circularity, we can do better.
This has effects on our linkage construction: Curves of high
circularity lead to minimal motions of low degree and low spherical
degree defect and lend themselves well to a realization by a linkage
with few links and joints.

The constructive proof of \autoref{th:2} in \cite{li15a} can be turned
into a short algorithm to actually compute minimal motions. We make a
few technical assumptions:
\begin{itemize}
\item The reduced rational curve $x = x_0+x_1\qi + x_2\qj + x_3\qk$
  satisfies $x(\infty) \coloneqq \lim_{t \to \infty} t^{-d}x = 1$.
  This can always be accomplished by a suitable translation of the
  coordinate frame.
\item We want to find a motion polynomial $C = P + \eps Q$ of minimal
  degree $d - c$ such that $P\cj{P} + 2P\cj{Q} = x$. In view of
  \eqref{eq:1} this means that $x$ is the trajectory of the affine
  origin in~$\R^3$.
\item We assume that $C$ is monic which is consistent with the
  assumption $x(\infty) = 1$ and entails that $C(\infty)$ is the
  identity. This can always be accomplished by a suitable rotation of
  the coordinate frame about its origin.
\end{itemize}

An important ingredient in the computation of minimal motions is right
division of quaternion polynomials. Given $F$, $G \in \H[t]$, there
exist unique polynomials $Q, R \in \H[t]$, called right quotient and
right remainder, with $F = GQ + R$ and $\deg R < \deg G$. In case of
monic $G$, they can be computed by \autoref{alg:rQR}. We denote the
right quotient by $Q = \rquo(F, G)$ and the right remainder by
$R = \rrem(F, G)$. The latter is used in \autoref{alg:lgcd} (Euclidean
algorithm) for computing the left gcd of two quaternion polynomials
$F, G \in \H[t]$ in case of monic $G$. The left gcd is the unique
monic polynomial $L = \lgcd(F,G)$ of maximal degree such that there
exist polynomials $Q, R \in \H[t]$ with $F = LQ$ and $G = LR$. The
function $\lcoeff$ in Line~\ref{line:lcoeff} of \autoref{alg:lgcd}
returns the leading coefficient of a polynomial so that
$\lcoeff(R)^{-1}R$ is monic.

\begin{algorithm}
  \caption{$\rQR(F, G)$ (quotient and remainder of polynomial right division)}
  \label{alg:rQR}
  \begin{algorithmic}[1]
    \Require Two polynomials $F, G \in \H[t]$, $G$ is monic.
    \Ensure Polynomials $Q, R \in \H[t]$ such that $\deg R < \deg G$
    and $F = GQ + R$.
    \State $Q \leftarrow 0$, $R \leftarrow F$
    \State $m \leftarrow \deg F$, $n \leftarrow \deg G$
    \While{$m \ge n$}
    \State $c \leftarrow \lcoeff(R)$ \Comment leading coefficient of $R$
    \State $Q \leftarrow Q + ct^{m-n}$
    \State \label{line:Gct}$R \leftarrow R - Gct^{m-n}$
    \State $m \leftarrow \deg R$
    \EndWhile
    \State \Return $Q, R$
  \end{algorithmic}
\end{algorithm}

\begin{algorithm}
  \caption{$\lgcd(F, G)$ (left gcd of quaternion polynomials)}
  \label{alg:lgcd}
  \begin{algorithmic}[1]
    \Require Two polynomials $F$, $G \in \H[t]$, $G$ is monic.
    \Ensure Monic polynomial $L \in \H[t]$ such that 
    there exist polynomial $Q, R \in \H[t]$ with $F = LQ$ and $G = LR$.
    \State $R \leftarrow \rrem(F, G)$
    \If{$R = 0$}
    \State \Return $G$
    \EndIf
    \State \label{line:lcoeff}\Return $\lgcd(G, R \lcoeff(R)^{-1})$
  \end{algorithmic}
\end{algorithm}

The computation of the minimal degree rational motion with the
prescribed trajectory $x$ is illustrated in \autoref{alg:minmot}. It
computes a monic rational motion polynomial $C$ such that the
trajectory of the affine origin is parameterized by $x$. The
correctness of \autoref{alg:minmot} has been proved in
\cite{li15a}. The function $\quo$ in Line~\ref{line:quo} denotes
the quotient of polynomial division for real polynomials. It may
also be computed by \autoref{alg:rQR}.

\begin{algorithm}
  \caption{$\minmot(x)$ (minimal degree rational motion; \cite{li15a})}
  \label{alg:minmot}
  \begin{algorithmic}[1]
    \Require Reduced rational parametric equation
    $x = x_0 + x_1\qi + x_2\qj + x_3\qk$ with $x(\infty) = 1$.
    \Ensure Monic motion polynomial $C$ of minimal degree such that
    $x = P\cj{P} + 2P\cj{Q}$ (trajectory of affine origin is
    parameterized by $x$).
    \State $g \leftarrow \gcd(x_0, x_1^2+x_2^2+x_3^2)$
    \State \label{line:quo}$w \leftarrow \quo(x_0,g)$ \Comment
    quotient of polynomial division in $\R[t]$ \State
    $D \leftarrow x_1\qi + x_2\qj + x_3\qk$
    \State $P' \leftarrow \lgcd(D,g)$, $Q' \leftarrow \rquo(D,P')$,
    \State \Return $C = wP' + \frac{1}{2} \eps \cj{Q'}$
  \end{algorithmic}
\end{algorithm}

\begin{lemma}
  \label{lem:1}
  Let $C=P+\eps Q$ be a minimal motion to the reduced rational curve
  $x = x_0 + x_1\qi + x_2\qj + x_3\qk$. Then $\mrpf(P)$ and $Q\cj{Q}$
  are relatively prime.
\end{lemma}

\begin{proof}
  Assume that the trajectory of the origin is
  $x_0+x_1\qi+x_2\qj+x_3\qk$. Following Algorithm~\ref{alg:minmot}, we
  define $g \coloneqq \gcd(x_0,x_1^2+x_2^2+x_3^2)$ and
  $D \coloneqq x_1\qi + x_2\qj + x_3\qk$. Then
  \begin{equation*}
    P = wP' = \frac{x_0}{g}\lgcd(D, g).
  \end{equation*}
  Because $x$ is reduced, $\gcd(x_1\qi + x_2\qj + x_3\qk,g) = 1$ and
  $\mrpf(P)=x_0/g$. Moreover, Lemma~3 of \cite{li15a} is applicable
  (with $C = D$ and $R = g$) and guarantees that
  $L \coloneqq \lgcd(D, g)$ satisfies $L\cj{L} = g$. There exists
  $D' \in \H[t]$ such that $D = LD'$. But then
  $-D = \cj{D} = \cj{D'}\,\cj{L}$ which implies $\cj{L} = \rgcd(D,g)$
  and also $\cj{Q} = \frac{1}{2}\lquo(D,\cj{L}) = \frac{1}{2}D'$.
  Thus,
  \begin{equation*}
    -4Q\cj{Q} = D'\cj{D'} = D\cj{D}(L\cj{L})^{-1} = D\cj{D}/g
  \end{equation*}
  and $\mrpf(P)=x_0/g$ and $Q\cj{Q}=-(x_1^2+x_2^2+x_3^2)/(4g)$ are
  relatively prime.
\end{proof}

\begin{remark}
  If the rational curve is planar or spherical, then the minimal
  motion is also planar or spherical, respectively. If a planar curve
  with a non-planar minimal motion existed, we could reflect the
  motion in the curve's plane and obtain a contradiction to
  uniqueness. For spherical curves, this follows from Folgerung~9 and
  the proof of Satz~6 of \cite{juettler93}.
\end{remark}

Let us look at some examples of minimal motions:

\begin{example}
  \label{ex:1}
  Consider the parametric curve $x = t^2 + 1 - 2 a\qi -2 b\qj t$ with
  $a, b \in \R$. For $a > b > 0$, it is an ellipse and the minimal
  motion polynomial is $C = t^2 + 1 + \eps(a\qi + b\qj t)$. It
  parameterizes the translation along the ellipse. If $a > b = 0$, the
  curve degenerates to a straight line segment and the minimal motion
  is the translation $C = t^2 + 1 + a \eps\qi$ (back and forth along
  this segment). If $a = b > 0$, the parametric curve is a circle. Its
  circularity is one and the minimal motion polynomial
  $C = t - \qk + a\eps\qj$ is linear, as predicted by \autoref{th:1}.
  It parameterizes the rotation around the circle axis.
\end{example}

\begin{example}
  \label{ex:2}
  As a second example, consider Viviani's curve $x = x_0 + x_1\qi +
  x_2\qj + x_3\qk$, given by
  \begin{equation*}
    x_0 = (1+t^2)^2,\quad
    x_1 = -4t^2,\quad
    x_2 = 2t(1-t^2),\quad
    x_3 = 2t(1+t^2).
  \end{equation*}
  It lies on the sphere of radius $1$ with center $(-1,0,0)$ and our
  parameterization satisfies $x(\infty) = 1$. The minimal motion is
  $C = t^2 - (\qj + \qk - \eps(\qj - \qk))t - \qi$. It is a spherical
  motion because it fixes the sphere center $(-1,0,0)$. We call it
  \emph{Viviani motion.}
\end{example}

\subsection{Factorization of the bounded minimal degree motion}
\label{sec:factorization}

Having constructed a monic motion polynomial $C$ with a trajectory $x$, we have
converted our trajectory generation problem to a motion generation problem: We
are looking for a linkage with one degree of freedom such that one link follows
the motion parameterized by $C$. (Note that we are only interested in bounded
motions, that is, motions with only bounded trajectories.) This we accomplish by
decomposing the motion into the product of rotations about certain axes. The
rotation angles are linked by the common motion parameter $t$. The basic tool
for this is the factorization theory for motion polynomials as introduced in
\cite{hegedus13} and its extension to non-generic bounded motion polynomials
in~\cite{li15b}.

\begin{definition}
  \label{def:factorization}
  A bounded motion polynomial $C$ is said to \emph{admit a
    factorization} if there exist bounded linear motion polynomials
  $t-h_1$, \ldots, $t-h_n$ such that
  \begin{equation}
    \label{eq:2}
    C = (t-h_1) \cdots (t-h_n).
  \end{equation}
\end{definition}

We already said that each linear factor $t - h_i$ in \eqref{eq:2}
parameterizes a rotation with fixed axes or translation in fixed
direction. Because $t - h_i$ is bounded, the later cannot occur
here. Hence, the product \eqref{eq:2} parameterizes the composition of
such rotations. Note that a factorization of the form \eqref{eq:2}
need not exist and if it exists, it need not be unique.

Call a motion polynomial ``generic'' if its spherical degree defect
(\autoref{def:trajectory-degree}) is zero.  A generic motion
polynomial can always be factored \cite[Theorem~1]{hegedus13} but the
factorization is in general not unique. Elliptic and circular translation
and the Viviani motion serve as examples:

\begin{example}
  \label{ex:3}
  The elliptic translation of \autoref{ex:1} does not admit a
  factorization if $a \neq b$. If $a = b > 0$, it admits infinitely
  many factorizations
  \begin{equation*}
    C = t^2 + 1 + a\eps(\qi + \qj t)
      = (t - \qk + \eps(\alpha\qi + (\beta+a)\qj))
        (t + \qk - \eps(\alpha\qi + \beta\qj))
  \end{equation*}
  with $\alpha$, $\beta \in \R$. The Viviani motion in \autoref{ex:2}
  admits only the factorization
  \begin{equation}
    \label{eq:3}
    C = t^2 - (\qj + \qk - \eps(\qj + \qk))t - \qi
      = (t - \qk + \eps\qj)(t - \qj - \eps\qk).
  \end{equation}
  By considering primal parts only (the dual parts are merely there
  because the sphere center is not the affine origin), we see that it
  is the composition of rotations with equal angular speed about the
  second and the third coordinate axis.
\end{example}

An algorithm for computing factorizations of generic motion
polynomials has been presented in \cite{hegedus13}. It is displayed in
\autoref{alg:czero} and \autoref{alg:gfactor}. The assumptions on $M$
and $C$ in \autoref{alg:czero} guarantee that the remainder $R$ in
Line~1 has an invertible leading coefficient (compare
\cite[Theorem~3]{hegedus13}). They are met in \autoref{alg:gfactor}
because $C$ is assumed to be generic. The non-uniqueness of the
factorization comes from the undetermined order of the quadratic
factors in Line~\ref{line:quadfac} of \autoref{alg:gfactor}. In
Line~\ref{line:lquo} we compute the quotient of polynomial left
division. This can be done by a variant of \autoref{alg:rQR} but with
$R-Gct^{m-n}$ in Line~\ref{line:Gct} replaced by $R-ct^{m-n}G$.

\begin{algorithm}
  \caption{$\czero(C,M)$ (common zero of $C$ and quadratic factor $M$
    of $C\cj{C}$)}
  \label{alg:czero}
  \begin{algorithmic}[1]
    \Require Monic, bounded motion polynomial $C \in \D\H[t]$,
    quadratic factor $M$ of $C\cj{C}$ that does not divide the primal
    part of~$C$.
    \Ensure Bounded linear motion polynomial $t - h$ such that
    $C(h) = M(h) = 0$.
    \State $R \leftarrow \rrem(C,M)$
    \Comment $R = at + b$ with $a,b \in \D\H$.
    \State \label{li:zero}$h \leftarrow $ unique zero of $R$ \Comment
    $h = -a^{-1}b$ ($a$ is invertible) \State \Return $h$.
  \end{algorithmic}
\end{algorithm}

\begin{algorithm}
  \caption{$\gfactor(C)$ (factorization of generic motion polynomials; \cite{hegedus13})}
  \label{alg:gfactor}
  \begin{algorithmic}[1]
    \Require $C = P + \eps Q \in \D\H[t]$, a generic, monic motion polynomial.
    \Ensure A list $L=[t-h_1,\ldots,t-h_n]$ of bounded linear motion
    polynomials such that $C=(t-h_1) \cdots (t-h_n)$.
    \State $L \leftarrow [\,]$
    \Comment initialize empty list
    \State \label{line:quadfac}$F \leftarrow [M_1,\ldots,M_n]$
    \Comment Each $M_i \in \R[t]$, $i\in\{1,\ldots,n\}$ is a\\
    \hfill quadratic, irreducible factor of $C\cj{C} \in \R[t]$.
    \For{$i = 1$ to $n$}
    \State $F \leftarrow \lmin{F}{M_i}$.
    \Comment Remove $M_i$ from list $F$.
    \State \label{line:czero}$h \leftarrow \czero(C, M_i)$
    \State $L \leftarrow \lplus{[t-h]}{L}$
    \Comment Add $t-h$ to start of list $L$.
    \State \label{line:lquo}$C \leftarrow \lquo(C,t-h)$
    \Comment Quotient of left division, variant of \autoref{alg:rQR}.
    \EndFor
    \State \Return $L$.
  \end{algorithmic}
\end{algorithm}

\begin{example}
  \label{ex:4}
  \autoref{alg:gfactor} cannot be used to factor the elliptic or
  circular translation of Examples~\ref{ex:1} and \ref{ex:3} because
  its primal part $t^2+1$ is real. It fails in Line~\ref{li:zero} of
  \autoref{alg:czero} where the dual quaternion $a$ is not
  invertible. We can, however, use \autoref{alg:gfactor} to obtain the
  factorization \eqref{eq:3} of the Viviani motion. Because of
  $C\cj{C} = (1+t^2)^2$ the quadratic factors of $C\cj{C}$ are all
  equal ($M_1 = M_2 = 1+t^2$) and the output of \autoref{alg:gfactor}
  is, indeed, unique.
\end{example}

Our linkage construction requires existence of at least one
factorization of the motion polynomial $C$. Thus, we have to find a
way to ``factorize'' non-generic motion polynomials as well. One
possibility to do this has been presented in \cite{li15b}. There, we
showed that for every bounded motion polynomial $C = P + \eps Q$ of
degree $n$ a polynomial $R \in \R[t]$ of degree $m \le \deg\mrpf(P)$
exists such that $CR$ admits a factorization. This statement for
planar motion polynomials has already been proved in \cite{gallet15}.
Although not reduced, the motion polynomial $CR$ parameterizes the
same motion as $C$. Thus, at the cost of possibly doubling the degree
of $C$, we can find a factorized representation to base our linkage
construction upon. In fact, the worst case $m = n$ occurs for
translational motions, that is, for generic (non-circular)
trajectories.

Multiplication with $R \in \R[t]$ changes the algebraic properties of
the motion polynomial but not the motion itself. However, for our
purpose it is sufficient to preserve just one particular
trajectory. Therefore, the following refinement is conceivable. If
$H \in \H[t]$ is a quaternion polynomial, the orbit of the affine
origin with respect to $CH$ is the same as the orbit with respect to
$C$. Indeed, if $C = P + \eps Q$, this orbit is parameterized by
$P\cj{P} + 2P\cj{Q}$ but, because of $H\cj{H} \in \R[t]$, also by
\begin{equation*}
  PH\cj{PH} + 2PH\cj{QH} =
  P(H\cj{H})\cj{P} + 2P(H\cj{H})\cj{Q} =
  (H\cj{H})(P\cj{P} + 2P\cj{Q}).
\end{equation*}
In order to retain the orbit of a point different from the affine
origin, we can similarly right-multiply with a motion polynomial that
fixes that point. Besides the extension to three-dimensions, a major
contribution in this article over \cite{gallet15} is to base our
linkage construction on factorizations of $CH$ instead of $CR$, thus
improving the bounds on the number of links and joints, even in the
planar case. But before doing so we have to answer three questions:
\begin{itemize}
\item Given a motion polynomial $C$ of minimal degree to a bounded trajectory, does there always exist a
  quaternion polynomial $H$ such that $CH$ admits a factorization?
\item Assuming a positive answer to the previous question: How can we determine
  $H \in \H[t]$?
\item Finally, how can we compute the factorization of $CH$. Note that
  \autoref{alg:gfactor} will not do unless it can already be used to
  factorize $C$ (in which case we simply have $H = 1$).
\end{itemize}
The first question will be positively answered in \autoref{th:3},
below. Its proof answers the second and third question.

A polynomial $P \in \H[t]$ of
degree $n$ and without non-trivial real  factors can have at most $n$
quaternion zeros. If $P$ has real factors, the situation is slightly
more complicated. However, with the boundedness condition, we will only encounter the simpler
situation that $\mrpf(P)$ is a \emph{strictly positive} real
polynomial, i.e., it is irreducible over $\R$. Then the quaternion
zeros of $\mrpf(P)$ are precisely the quaternion zeros of the
irreducible, real, quadratic factors of $\mrpf(P)$ and these are well
known:

\begin{lemma}[\cite{huang02}]
  \label{lem:2}
  If the real polynomial $R = t^2 + bt + c$ is irreducible over $\R$
  ($4c - b^2 > 0$), the set of its quaternion zeros is
  \begin{equation*}
    \{\tfrac{1}{2}(-b + s_1\qi + s_2\qj + s_3\qk) \mid (s_1,s_2,s_3)\in\R^3,\ s_1^2+s_2^2+s_3^2=4c-b^2 \}.
  \end{equation*}
\end{lemma}

One consequence of \autoref{lem:2} is the existence of a unique zero
whose vector part (the projection onto
$\langle \qi, \qj, \qk \rangle$) equals a positive multiple of an
arbitrary non-zero vector $v \in \R^3$. We call this a zero in
\emph{direction of $v$.}

We say that a motion polynomial $C=P+\eps Q$ is \emph{tame} if
$\mrpf(P)$ and $Q\cj{Q}$ are relatively prime. Note that minimal
motions are always tame by Lemma~\ref{lem:1}. The following lemma
allows to reduce the degree of $\mrpf(P)$ for tame polynomials.

\begin{lemma}
  \label{lem:3}
  Consider a tame, monic, and bounded motion polynomial $C = P + \eps Q$ and
  denote by $F \in \R[t]$ a monic, quadratic and irreducible factor of $p
  \coloneqq \mrpf(P)$. Then there exists a tame, monic, and bounded motion
  polynomial $C' = P' + \eps Q'$ such that $\mrpf(P') = \mrpf(P)/F$ and linear
  quaternion polynomials $L$, $L'$ such that $CL = L'C'$.
\end{lemma}

\begin{proof}
  By Lemma~\ref{lem:2}, we can write $F=L\cj{L}$ for some monic linear
  quaternion polynomial $L = t - h$. If $F$ divides $QL$, then $F^2$ divides
  $Q\cj{Q}L\cj{L}=Q\cj{Q}F$. But then $F$ would divide $Q\cj{Q}$ which
  contradicts the tameness assumption. So, $F$ does not divide $QL$ and we can
  use \autoref{alg:czero} to compute a common zero $\cj{h'}$ of $\cj{QL}$ and
  $F$. Moreover, it is no loss of generality to assume that $t-\cj{h}$ avoids the
  finitely many right factors of $P/p$ and $t-h'$ avoids the finitely many
  left factors of $P/p$. Setting $L' \coloneqq t-h'$ and
  $Q' \coloneqq \rquo(QL, L')$,
  we obtain $QL = L'Q'$. Then
  \begin{equation*}
    CL = PL + \eps QL
       = L'\cj{L'}(P/F)L + \eps L'Q'
       = L'(\cj{L'}(P/F)L + \eps Q'),
  \end{equation*}
  and we set $C' \coloneqq P' + \eps Q'$ where $P' \coloneqq \cj{L'}(P/F)L$. By
  a standard result on quaternion polynomials $\mrpf(P')=\mrpf(P/F)$, because
  additional real factors only arise in products $P_1P_2$ when a linear right
  factor of $P_1$ is conjugate to a linear left factor of $P_2$
  (\cite[Lemma~1]{li15a} or \cite[Proposition~2.1]{cheng16}). Finally, $C'$ is
  also tame and bounded, because the norm of the dual part has not changed and
  the minimal real polynomial factor only got smaller.
\end{proof}

\begin{theorem}
  \label{th:3}
  Given a tame monic bounded motion polynomial $C = P + \eps Q \in \DH[t]$
  there exists a polynomial $H \in \H[t]$ 
  of degree
  \begin{equation*}
    \deg H = \tfrac{1}{2}\deg\mrpf(P)
  \end{equation*}
  such that $CH$ admits a factorization.
\end{theorem}

\begin{proof}
  We proceed by induction on $n \coloneqq \deg\mrpf(P)$.

  If $n=0$, then $C$ is generic and can be factored by
  Algorithm~\ref{alg:gfactor}.

  Assume that $\deg\mrpf(P)>0$ and let $F$ be a monic, irreducible,
  quadratic factor of $\mrpf(P)$. By Lemma~\ref{lem:3}, there
  exist linear polynomials $L$, $L' \in \H[t]$ and a tame motion
  polynomial $C' = P' + \eps Q'$ such that $L'C'=CL$ and
  $\deg\mrpf(P')=n-2$. By induction hypothesis, there exists a
  quaternion polynomial $H'$ of degree  $(n-2)/2$ such that
  $C'H'$ can be factorized, say $C'H'=M_1 \cdots M_{d}$ for suitable
  rotation polynomials $M_1$,\ldots,$M_d$. Then we set
  $H \coloneqq LH'$ and we have a factorization
  $CH = CLH' = L'C'H' = L'M_1\cdots M_d$.
\end{proof}

\autoref{alg:tfactor} displays pseudocode, derived from the proofs of
\autoref{lem:3} and \autoref{th:3}, to compute a factorization for a
tame motion polynomial $C$. It returns a list $[L_1,\ldots,L_n]$ of
linear motion polynomials and $H \in \H[t]$ such that $L_1 \cdots L_n = CH$.
The correctness is clear from the proof of \autoref{lem:3}.

\begin{algorithm}
  \caption{$\tfactor(C)$ (Factorization of tame motion polynomials)}
  \label{alg:tfactor}
  \begin{algorithmic}[1]
    \Require Tame motion polynomial $C = P + \eps Q \in \DH[t]$.
    \Ensure Pair $(M,H)$ consisting of a list $M = [L_1,\dots,L_n]$ of linear
    motion polynomials and a quaternion polynomial $H$ such that $CH = L_1 \cdots L_n$.
    \If{$\mrpf(P)=1$}
      \State \Return $(\gfactor(C), 1)$
    \Else
      \State Choose an irreducible quadratic factor $F$ of~$\mrpf(P)$.
      \State Choose a random zero $h\in\H$ of $F$ and set
       $L \leftarrow t-h$. \label{line:tfactor-randzero}
      \State $E \leftarrow \cj{QL}$
      \State $h'\leftarrow \cj{\czero(E, F)}$
      \State $L'\leftarrow t - h'$
      \State $P'\leftarrow \cj{L'}\tfrac{P}{F}L$, $Q'\leftarrow \rquo(\cj{E}, L')$
      \State $C' \leftarrow P'+\eps Q'$
      \State $M',H'\leftarrow \tfactor(C')$
      \State \Return $(\lplus{L'}{M'}, LH')$
    \EndIf
  \end{algorithmic}
\end{algorithm}

\begin{remark} In Line~\ref{line:tfactor-randzero} of \autoref{alg:tfactor} we
  are free to pick one among infinitely many quaternion zeros of an irreducible
  quadratic polynomial. We only have to avoid finitely many ``dangourous'' zeros
  of left or right factors of $P/p$, as mentioned in the proof of
  \autoref{lem:3}. This freedom is quite advantageous for applications of our
  algorithm in order to fulfill engineering needs. In the planar case, one has
  to pick the suitable one among two conjugate solutions. This is always
  possible because $P$ is in a sub-algebra isomorphic to $\C$. Hence, we need
  not distinguish between left and right factors and no two linear factors of
  $P/p$ are conjugate.
\end{remark}

\begin{example}
  \label{ex:5}
  We continue \autoref{ex:1} and illustrate \autoref{alg:tfactor} at
  hand of the elliptic translation $C = t^2 + 1 + \eps(b\qj t + a\qi)$
  with $a > b \ge 0$. We have $f = t^2+1$ and may choose $h = \qi$
  and $L = t - \qi$ in Line~\ref{line:tfactor-randzero}. With
  $P = t^2 + 1$ and $Q = b\qj t + a\qi$ this yields
  \begin{equation*}
    E = \cj{Q L} = -b\qj t^2 - (a\qi + b\qk)t + a,\quad
    h' = \frac{(a^2-b^2)\qi + 2ab\qk}{a^2+b^2}.
  \end{equation*}
  The algorithm recursively calls itself with input $C' = P' + \eps Q'$ where
  \begin{gather*}
    L' = t - h',\quad
    P' = \frac{\cj{L'}PL}{f} = t^2 + \frac{1}{a^2+b^2}\bigl( 2b(a\qk - b\qi)t + a^2 - 2ab\qj - b^2 \bigr),\\
    Q' = \rquo(\cj{E}, L') = b\qj t + \frac{a}{a^2+b^2}\bigl( (a^2-b^2)\qi + 2ab\qk \bigr).
  \end{gather*}
  It admits the factorization $C' = (t-k_1)(t-k_2)$ where
  \begin{equation*}
    k_1 = -\frac{(a^2-b^2)\qi + 2ab\qk}{a^2+b^2}
          - \eps\qj\frac{a^2+b^2}{2b},
    \quad
    k_2 = \qi + \eps\qj \frac{a^2-b^2}{2b}.
  \end{equation*}
  This gives the factorization 
  \begin{equation}
    \label{eq:4}
      CH = (t-h')(t-k_1)(t-k_2).
  \end{equation}
  The quaternion polynomial factor equals $H = t - \qi$.

  Note that above factorization is spatial even if the elliptic translation is a
  planar motion. In order to obtain a planar factorization, we may select $h =
  \qk$ and $L = t - \qk$ in Line~\ref{line:tfactor-randzero}. This choice gives
  a simple planar factorization:
  \begin{equation*}
    E = -b\qj t^2 - (a-b)\qi t - a \qj,\quad
    L' = t + \qk,\quad
    C' = t^2 - 2\qk t - 1 + \eps (b\qj t + a\qi)
  \end{equation*}
  and
  \begin{equation}\label{eq:5}
    CH = (t + \qk)(t-\qk - \tfrac{1}{2}\eps\qj(a - b))(t-\qk + \tfrac{1}{2}\eps\qj(a + b))
  \end{equation}
  where $H = t - \qk$.
\end{example}

\begin{remark}
  One can show that all trajectories of the motion $CH$ in
  \autoref{ex:5} are ellipses (or line segments). In the spatial case,
  the ellipses' planes are not all parallel. This is a characterizing
  property of the \emph{Darboux-motion} whose factorizations have
  already been discussed in \cite{li15c}.
\end{remark}

\subsection{Bennett flips}
\label{sec:bennett-flips}

In \autoref{sec:factorization} we showed how to factor a bounded
motion polynomial, possibly after right multiplication with a
quaternion polynomial, into the product of bounded linear motion
polynomials. This section is a little intermezzo before constructing
linkages from these factorizations in
\autoref{sec:scissor-linkage}. It introduces a technique we call
``Bennett flip''. It will be an essential component of our linkage
construction.

In general, a motion polynomial admits several factorizations, each
giving rise to an open chain of revolute joints that is capable of
performing the motion parameterized by $C$. The distal joints of
different factorizations can be attached to a common link and the
resulting multi-looped linkage can still perform the motion
$C$. Adding sufficiently many factorizations one can reasonably expect
to reduce the dimension of the linkages configuration space to
one. This idea has already been successfully applied for constructing
linkages for engineering applications
\cite{hegedus15:_four_pose_synthesis} but it is not really suitable
for proving the general statement of \autoref{th:1} because it may
fail to produce linkages with only one degree of freedom. Consider,
as a warning, the Viviani motion \eqref{eq:3}. It admits only one
factorization and hence gives raise to just one single chain with two
revolute joints and two degrees of freedom.

It is, however, possible to use the above idea in the special case of
quadratic motion polynomials. They are sufficiently simple to allow a
complete discussion of all unwanted cases. Once this is available we
can combine degree two motion polynomials and their linkages to
generate mechanisms that satisfy the criteria of \autoref{th:1}.

\begin{definition}
  \label{def:bflip}
  The \emph{Bennett flip} is the map
  \begin{equation*}
      \bflip\colon \DH^2 \setminus \{ (h_1,h_2) \mid \cj{h_1} = h_2\} \to \DH^2,\quad
      (h_1,h_2) \mapsto (k_1,k_2).
  \end{equation*}
  where $k_2 \coloneqq -(\cj{h_1}-h_2)^{-1}(h_1h_2 - h_1\cj{h_1})$ and
  $k_1 \coloneqq h_1 + h_2 - k_2$.
\end{definition}

In order to understand the idea behind \autoref{def:bflip}, take two rotation
quaternions $h_1$, $h_2$ with $\cj{h_1} \neq h_2$, set $(k_1,k_2) \coloneqq
\bflip(h_1,h_2)$ and consider the polynomials $C = (t-h_1)(t-h_2)$, $C\cj{C} =
M_1M_2$ where $M_1 = (t-h_1)(t-\cj{h_1})$ and $M_2 = (t-h_2)(t-\cj{h_2})$. From
\begin{equation*}
  (t-h_1)(t-h_2) = (t-h_1)(t-\cj{h_1}) + (\cj{h_1}-h_2)t + h_1h_2 - h_1\cj{h_1}
\end{equation*}
we see that $k_2$ is the zero of $\rrem(C,M_1)$ and $k_1$ is the zero of
$\lquo(C,t-k_2) = t - k_1$. In other words, $k_1$ and $k_2$ are obtained by
applying \autoref{alg:gfactor} to $C = (t-h_1)(t-h_2)$ and yield the second
factorization $C = (t-k_1)(t-k_2)$. This interpretation accounts for the name
``Bennett flip'': Given a rotation quaternion $h$, denote its axes by
$\axis{h}$. In general the axes $\axis{h_1}$, $\axis{h_2}$, $\axis{k_2}$, and
$\axis{k_1}$ form a Bennett linkage. Exceptional cases exist and will be
described in detail.

\begin{proposition}
  \label{prop:bflip}
  Consider two rotation quaternions $h_1$, $h_2 \in \DH$ with
  $\cj{h_1} \neq h_2$ and let $(k_1,k_2) \coloneqq
  \bflip(h_1,h_2)$. Then the following hold:
  \begin{enumerate}
  \item Also $k_1$ and $k_2$ are rotation quaternions.
  \item The restriction of the Bennett flip to pairs of rotation
    quaternions is an involution, that is,
    $\bflip(k_1,k_2) = (h_1,h_2)$.
  \item The restriction of the Bennett flip to pairs of rotation
    quaternions is birational.
  \item We have $\minpol(k_1) = \minpol(h_2)$ and $\minpol(k_2) =
    \minpol(h_1)$ where $\minpol(h) \coloneqq (t-h)(t-\cj{h})$ is the
    \emph{minimal polynomial} of $h \in \DH$.
  \item We have
    $\bflip(h_2,\cj{k_2}) = (\cj{h_1},k_1)$,
    $\bflip(\cj{k_2},\cj{k_1}) = (\cj{h_2},\cj{h_1})$ and
    $\bflip(\cj{k_1},h_1) = (k_2,\cj{h_2})$.
  \end{enumerate}
\end{proposition}

\begin{proof}
  Items~1 and 2 follow from the interpretation of the Bennett flip as
  application of \autoref{alg:gfactor} to $C = (t-h_1)(t-h_2)$.
  Item~2 implies Item~3. Item~4 follows again from
  \autoref{alg:gfactor} because $\minpol(h) = M_i$ in
  Line~\ref{line:czero} of that algorithm. In order to see the last
  item, we may multiply the equation $(t-h_1)(t-h_2) = (t-k_1)(t-k_2)$
  with $t-\cj{h_1}$ from the left and with $t-\cj{k_2}$ from the
  right. With
  $M_1 \coloneqq (t-h_1)(t-\cj{h_1}) = (t-k_2)(t-\cj{k_2})$ we obtain
  $M_1(t-h_2)(t-\cj{k_2}) = (t-\cj{h_1})(t-k_1)M_1$. Because $M_1$
  commutes with the other factors, the first equation of Item~5
  follows. The other equations follow by iterating this argument.
\end{proof}

\begin{remark}
  The statements of \autoref{prop:bflip} on the restriction of
  $\bflip$ to pairs of rotation quaternions are also true for $\bflip$
  itself. Here, we only need the weaker formulation that allows a
  simpler proof.
\end{remark}

With $(k_1,k_2) \coloneqq \bflip(h_1,h_2)$ we have
$(t-h_1)(t-h_2) = (t-k_1)(t-k_2)$. This ensures that the four-bar
linkage with axes $\axis{h_1}$, $\axis{h_2}$, $\axis{k_2}$, and
$\axis{k_1}$ moves with at least one degree of freedom. Moreover, it
is elementary to see that it moves with at most one degree of freedom,
if these axes are all different. If this is not the case, we have
$\axis{h_1} = \axis{k_1}$ and $\axis{h_2} = \axis{k_2}$ (this
comprises the case where all axes coincide). The following lemma
provides a sufficient criterion to exclude this.

\begin{lemma}
  \label{lem:4}
  Given are rotation quaternions $h_1$, $h_2$, $k_1$, $k_2$ with
  $(k_1,k_2)=\bflip(h_1,h_2)$.
  \begin{itemize}
  \item We have $\axis{h_1} = \axis{h_2}$ if and only if $h_1 -
    \cj{h_1}$ and $h_2 - \cj{h_2}$ are linearly dependent. In this
    case, $\axis{h_1} = \axis{h_2} = \axis{k_1} = \axis{k_2}$.
  \item Provided $\axis{h_1} \neq \axis{h_2}$, we have
    $\axis{h_1} = \axis{k_1}$ and $\axis{h_2} = \axis{k_2}$ if and
    only if $\minpol(h_1) = \minpol(h_2)$.
  \end{itemize}
\end{lemma}

\begin{proof}
  The dual quaternion $h_1 - \cj{h_1}$ has zero scalar part. Hence,
  there exist real numbers $l_1,\ldots,l_6$ such that
  $h_1 - \cj{h_1} = l_1\qi + l_2\qj + l_3\qk + \eps(l_4\qi + l_5\qj +
  l_6\qk)$
  and the Plücker line coordinates of $\axis{h_1}$ are
  $[l_1,l_2,l_3,-l_4,-l_5,-l_6]$. This implies the first claim.

  Now we turn to the second claim. By \autoref{alg:gfactor} it is
  obvious that $\minpol(h_1) = \minpol(h_2)$ implies equality of
  axes. Conversely, we have to show $h_i + \cj{h_i} = k_i + \cj{k_i}$
  and $h_i\cj{h_i} = k_i\cj{k_i}$ for $i \in \{1,2\}$ under the
  assumption of $\axis{h_1} = \axis{k_1}$ and
  $\axis{h_2} = \axis{k_2}$. This equality of axes implies linear
  dependence of their Plücker vectors. Hence, there exist real
  numbers $a_1$, $a_2$, $b_1$, $b_2$ such that $k_1 = a_1 + b_1h_1$
  and $k_2 = a_2 + b_2h_2$. By comparing coefficients on both sides of
  $(t-h_1)(t-h_2) = (t-k_1)(t-k_2)$ we obtain
  \begin{equation*}
    (b_1-1)h_1 + (b_2-1)h_2 + a_1 + a_2 = (b_1b_2-1)h_1h_2 + a_1a_2 = 0.
  \end{equation*}
  The first equation minus its conjugate is
  $(b_1-1)(h_1-\cj{h_1}) + (b_2-1)(h_2-\cj{h_2}) = 0$. Because
  $h_1-\cj{h_1}$ and $h_2-\cj{h_2}$ are linearly independent, this
  implies $b_1 = b_2 = 1$. But then, we are left with
  $a_1+a_2 = a_1a_2 =0$ and $a_1 = a_2 = 0$ follows. Thus, $h_1 = k_1$
  and $h_2 = k_2$. But then \autoref{prop:bflip}.4 implies
  $\minpol(h_1) = \minpol(k_2) = \minpol(h_2)$.
\end{proof}

\begin{remark}
  \label{rem:flip-mobility}
  Summarizing the result of \autoref{lem:4}, we can say that the
  four-bar linkage with axes $\axis{h_1}$, $\axis{h_2}$, $\axis{k_2}$,
  and $\axis{k_1}$ has exactly one degree of freedom if
  $h_1 - \cj{h_1}$ and $h_2 - \cj{h_2}$ are linearly independent and
  $\minpol(h_1) \neq \minpol(h_2)$.
\end{remark}

Under the assumptions of the previous remark, a four-bar linkage
obtained by a Bennett flip is of one of the following types:
\begin{itemize}
\item A Bennett linkage,
\item a planar anti-parallelogram linkage, or
\item a spherical linkage where the connections of opposite joint
  pairs are of equal length.
\end{itemize}
The spatial case is clear since Bennett linkages are the only movable
spatial four-bar linkages. In particular, the orthogonal distances and
angles of opposite axes pairs are equal. The statement on the planar
case has been proved in \cite{gallet15}. The spherical situation is
similar but a distinction between parallelogram and anti-parallelogram
linkages is neither possible nor necessary. The reason is that a joint
may equally well be realized by a circular arc or its supplementary
arc. Equal distance of joint pairs follows from projection of the
spatial case on the primal part, that is, by considering only the
spherical motion component. At any rate, the configuration curve of
planar or spherical linkages consists of two irreducible components
and only one of them is relevant to us.

\begin{remark}
  \label{rem:circular-flip}
  It is worth mentioning that techniques similar to Bennett flips are
  conceivable. For the circular translation $C = t^2+1+\eps \qi(t-\qj)$
  one can find infinitely many different factorizations and two of
  them may be combined, similar to the Bennett flip, to form a
  parallelogram linkage. This was used to construct new
  overconstrained 6R linkages in \cite{li14ck}.
\end{remark}

\subsection{Construction of a scissor linkage}
\label{sec:scissor-linkage}

Now we come to the actual construction of a linkage -- the last step in our
proof of \autoref{th:1}. In the preceding sections we have shown how to
\begin{enumerate}
\item compute a motion polynomial $C = P + \eps Q$ of minimal degree
  $d - c$ (where $d$ is the curve's degree and $c$ is its circularity)
  such that the trajectory of the affine origin equals the given trajectory
  and
\item determine a quaternion polynomial $H \in \H[t]$ of degree
  \begin{equation*}
    m \coloneqq \frac{1}{2}\deg\mrpf(P)
  \end{equation*}
  such that $CH$ admits the factorization $CH = (t-h_1) \cdots (t-h_n)$
  with rotation quaternions $h_1,\ldots,h_n \in \D\H$ and $n = d - c + m$.
\end{enumerate}
This factorization gives rise to an open chain of revolute axes that
can generate the motion parameterized by $CH$ in the following way.
\begin{itemize}
\item For $i \in \{1,\ldots,n\}$, the quaternion $h_i$ describes a
  rotation about the axis $\axis{h_i}$ whose Plücker coordinates
  $[l_1,\ldots,l_6]$ can be computed from
  $h_i - \cj{h_i} = l_1\qi + l_2\qj + l_3\qk - \eps(l_4\qi + l_5\qj +
  l_6\qk)$.
\item The lines $\axis{h_1},\ldots,\axis{h_n}$ determine the
  configuration of the linkage at parameter time $t = \infty$ (zero
  rotation angle). The configuration at parameter time $t \neq \infty$
  is obtained from this configuration by successively subjecting
  $\axis{h_{n-i+1}},\ldots,\axis{h_n}$ to the rotation $t-h_{n-i}$ for
  $i = 1,\ldots,n-1$.
\end{itemize}

This open revolute chain has $n$ degrees of freedom. It is our aim to
constrain its motion in such a way that a link attached to the last
joint (with axis $\axis{h_n}$) performs a motion that can be
parameterized by $CH$. The basic technique for doing this is the
Bennett flip of \autoref{sec:bennett-flips}. We can constrain the
motion of this chain to the motion parameterized by $CH$ by
\begin{enumerate}
\item picking a ``suitable'' motion polynomial $t - m_0$ and
\item recursively defining
  \begin{equation}
    \label{eq:6}
    (m_\ell, k_\ell) \coloneqq \bflip(h_\ell, m_{\ell-1}),\quad
    \text{for $\ell = 1,\ldots,n$}.
  \end{equation}
\end{enumerate}
This gives us rotation quaternions $h_1,\ldots,h_n$, $k_1,\ldots,k_n$,
and $m_0,\ldots,m_n$ that can be assembled to a linkage whose link
graph is depicted in \autoref{fig:linkgraph}. The vertices of the
linkgraph correspond to links and the edges correspond to joints. Two
vertices (links) are connected by an edge (joint) if relative motion
of the two links is constrained by the corresponding revolute
joint. The recursion \eqref{eq:6} builds the linkage from the bottom
row of edges (labeled $h_1$, \ldots, $h_n$) and the first vertical
edge (labeled $m_0$) but it is clear that we may equally well
start from any other vertical edge $m_i$, $i \in \{1,\ldots,n\}$.

\autoref{fig:linkgraph} also shows the joint hypergraph (not a graph because the
joint triples $(m_i,h_i,h_{i+1})$ and $(m_i,k_i,k_{i+1})$ belong to the same
links for $i \in \{1,\ldots,n-1\}$). The structure of this hypergraph suggests
the name ``scissor linkage'' for the resulting linkage type. The axes of $h_1$
and $m_0$ are attached to the fixed link, while $h_n$ and $m_n$ are attached to
the moving link. It performs the motion parameterized by the given motion
polynomial. In general, each loop $h_i$, $m_i$, $k_i$, $m_{i-1}$ forms a Bennett
linkage. In this case the linkage has just one degree of freedom and its
configuration curve has a single component. In other words, it satisfies all
requirements of \autoref{th:1}. It is, however, possible, that planar or
spherical four-bar linkages occur. This may be desirable, acceptable, or not
acceptable, depending on circumstances. In order to prove \autoref{th:1} we
must, however, avoid coinciding axes, that is, we have to show that one can pick
$m_0$ in such a way that the conditions of \autoref{lem:4} are never fulfilled.

\begin{figure}
  \centering
  \includegraphics{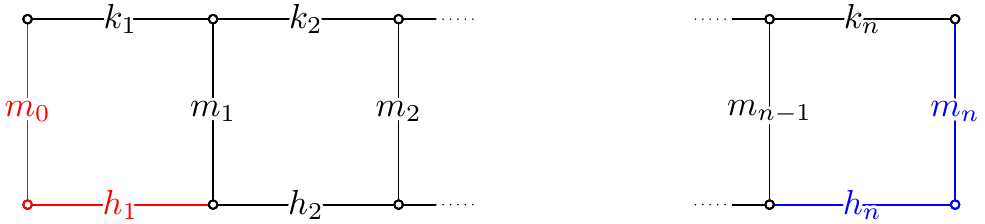}\\[1.0ex]
  \includegraphics{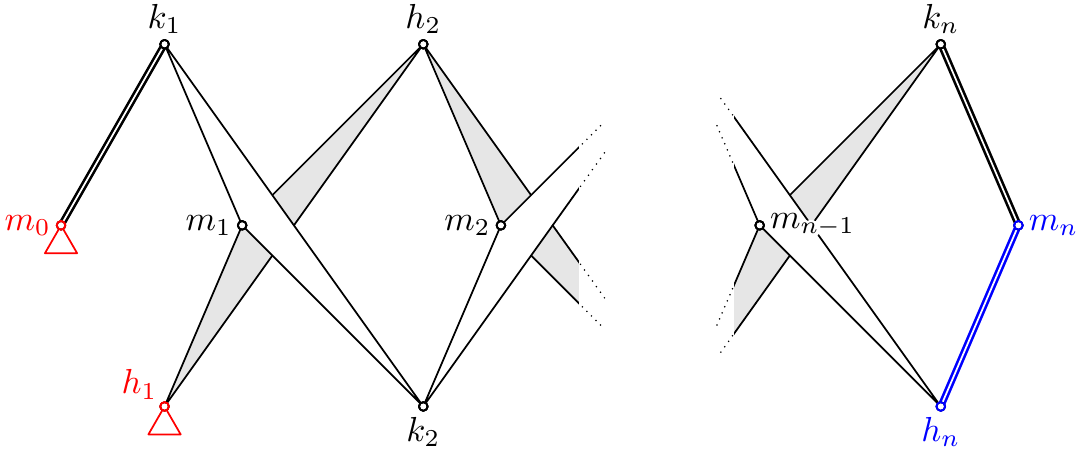}
  \caption{Link graph (top) and joint hypergraph (bottom) of scissor linkage}
  \label{fig:linkgraph}
\end{figure}

\begin{lemma}
  \label{lem:5}
  Given rotation quaternions $h_1,\ldots,h_n \in \DH$ there exists a
  rotation quaternion $m_0$ such that the quaternions $m_1,\ldots,m_n$
  and $k_1,\ldots,k_n$ obtained from the recursion \eqref{eq:6} are
  well-defined and the four-bar linkage with axes $\axis{m_{i-1}}$,
  $\axis{h_i}$, $\axis{m_i}$, $\axis{k_i}$ has precisely one degree of
  freedom for any $i \in \{1, \ldots, n\}$.
\end{lemma}

\begin{proof}
  By \autoref{lem:4}, we have to avoid linear dependency of
  $m_{i-1} - \cj{m_{i-1}}$ and $h_i-\cj{h_i}$ as well as equality
  $\minpol(m_{i-1}) = \minpol(h_i)$ of minimal polynomials. By
  \autoref{prop:bflip} we have $\minpol(m_0) = \ldots =
  \minpol(m_n)$. Thus, in order to ensure different minimal
  polynomials, we just have to avoid the set
  \begin{equation*}
    \bigcup_{i=1}^n
    \{ m \in \DH \mid m + \cj{m} = h_i + \cj{h_i},\ m\cj{m} = h_i\cj{h_i} \}
  \end{equation*}
  when picking $m_0$. This is the union of at most $n$ algebraic sets
  of positive codimension.

  In order to avoid linearly dependent vector parts, we recursively
  define maps
  \begin{equation*}
    \beta_0\colon m_0 \mapsto m_0,\quad
    \beta_{i-1}^{-1} \circ \beta_i\colon m_i \mapsto m_{i-1} :\iff \bflip(m_{i-1},h_i) = (k_i,m_i)
  \end{equation*}
  for $i \in \{1,\ldots,n\}$. They are well-defined birational maps by
  \autoref{prop:bflip} and satisfy $\beta_i(m_i) = m_0$.

  The vectors $m_i-\cj{m_i}$ and $h_i-\cj{h_i}$ are linearly dependent
  if the projection of $m_i$ onto the vector part is in the span of
  the projection of $h_i$ onto its vector part. This defines a vector
  subspace $L_i$ of dimension three that $m_i$ should avoid. This is
  the case if and only if $m_0$ avoids the union
  $\bigcup_{i=1}^n \beta_i(L_i)$, that is, another finite union of
  varieties of positive codimension. This is certainly possible.
\end{proof}

This finishes our proof of \autoref{th:1}. We just want to explain how
to bound the number of links and joints. The minimal degree motion
polynomial with a given rational curve of degree $d$ and circularity
$c$ as trajectory is of degree $d - c$ and spherical degree defect
$s = d - 2c$ (\autoref{th:2}). By \autoref{th:3}, there exist a
polynomial $H \in \H[t]$ of degree not larger than
$\frac{1}{2}s = \frac{1}{2}d - c$ such that $CH$ admits a
factorization. The product $CH$ is of degree
$n = d - c + \frac{1}{2}d - c = \frac{3}{2}d - 2c$ and this is the
same number $n$ as in \autoref{fig:linkgraph}. The numbers of links
and joints are the numbers of vertices and edges, respectively, in the
linkgraph, that is, $2(n+1) = 3d - 4c + 2$ and
$3n+1 = \frac{9}{2}d - 6c + 1$, respectively.

We conclude this section with a discussion of planar and spherical
four-bar linkages as components of the scissor linkage. The summary is
that they can be avoided if wanted but can also be enforced if the
given rational curve is planar or spherical.

\begin{corollary}
  \label{cor:1}
  Under the assumptions of \autoref{lem:5} we may pick $m_0$ in such a
  way that none of the quadruples $(m_{i-1}, h_i, m_i, k_i)$ is planar
  or spherical. In particular, our construction yields linkages such
  that the configuration space is free of spurious
  components.
\end{corollary}

\begin{proof}
  The quadruple will be spherical or planar, respectively, if and only
  if the axes of $m_{i-1}$ and $h_i$ intersect or are
  parallel. Similar as in the proof of \autoref{lem:5}, this defines
  subvariaties of positive codimension that $m_{i-1}$ should avoid for
  $i \in \{2,\ldots,n\}$. Mapping back this subvariaties via the maps
  $\beta_1,\ldots,\beta_n$ adds finitely many further components of
  positive codimension to the set that $m_0$ should avoid.
\end{proof}

\autoref{cor:1} shows that we can avoid planar or spherical
quadrilaterals in the scissor linkage if we wish. This ensures that
the configuration space of the resulting linkage is free from spurious
components. This is an important difference to planar versions of
Kempe's Universality Theorem where ``bracing constructions'' for
anti-parallelograms are necessary in order to suitable constrain the
configuration space. Of course, this then yields spatial linkages even
for planar or spherical curves/motions which may not always be desirable. The
next corollary states that it is also possible to generate planar or
spherical linkages if the input data is suitable.

\begin{corollary}
  \label{cor:2}
  If, under all assumptions of \autoref{lem:5}, the axes of
  $h_1,\ldots,h_n$ are incident with the same point or are
  parallel to the same direction, we may pick $m_0$ in such a way
  that all axes of $m_1,\ldots,m_n$ and $k_1,\ldots,k_n$ are incident
  with the this point or parallel to this direction,
  respectively.
\end{corollary}

\begin{proof}
  For the planar case, this is \cite[Lemma~6.5]{gallet15}. As to the
  spherical case, assume, without loss of generality that all dual
  parts are zero. The proof of \autoref{lem:5} shows that we can pick
  $m_0 \in \H$ which then implies $m_1,\ldots,m_n \in \H$ and the
  claim follows.
\end{proof}

Finally, we mention two ways to improve our bounds on the number of
links and joints. The first is just a hint that will often work for
spatial linkages without a formal proof. The second is a substantial
improvement of Kempe's Universality Theorem for rational spherical
curves.

\begin{remark}
  If two successive four-bars $\axis{h_i}$, $\axis{m_i}$, $\axis{k_i}$,
  $\axis{m_{i-1}}$ and $\axis{h_{i+1}}$, $\axis{m_{i+1}}$, $\axis{k_{i+1}}$,
  $\axis{m_i}$ in a scissor linkage are both Bennett linkages, it is, in
  general, possible to eliminate their common joint $m_i$ without increasing the
  mobility. The two Bennett linkages are replaced by a closed-loop linkage with
  six revolute joints and one degree of freedom. (It is of type ``Waldron's
  double Bennett hybrid'' \cite[pp.~63--65]{dietmaier95}.) We refrain from any
  attempts to formalizing this idea because it is a little tricky to guarantee
  that neither the number of irreducible components nor dimension of the
  configuration space increases in non-generic situations.
\end{remark}

\begin{corollary}
  \label{cor:spherical}
  A spherical rational curve of degree $d$ appears as trajectory of
  linkage with at most $d+2$ links and $\frac{3}{2}d+1$ joints.
\end{corollary}
\begin{proof}
  The statement follows from the bounds of \autoref{th:1} and the
  observation that spherical rational curves are of maximal
  circularity~$c = \frac{1}{2}d$.
\end{proof}

\section{Examples}
\label{sec:examples}

In this section we present several examples of linkages with given
rational curves. They are meant to illustrate type and properties of
the linkages we obtain, the degrees of freedom we have in our
construction, and how to use them in the design of linkages. Some of
them continue previous examples.

\begin{example}
  \label{ex:pascal}
  The limacon of Pascal, parameterized by $x = x_0 + x_1\qi + x_2\qj$ where
  \begin{equation*}
    x_0 = (1+t^2)^2,\quad
    x_1 = 2t(a-b-(a+b)t^2),\quad
    x_2 = (4a+2b)t^2+2b,
  \end{equation*}
  has the minimal motion
  \begin{equation*}
    C = (t - \qk + \tfrac{1}{2}\eps\qi(a+2b))(t - \qk + \tfrac{1}{2}\eps\qi a).
  \end{equation*}
  Because of $C\cj{C} = (1+t^2)^2$, this is the only factorization. We
  select $m_0 = 2\qk$ and, using the recursion \eqref{eq:6}, obtain
  \begin{gather*}
     3m_1 = 6\qk - 2\eps\qi(a+2b),\quad
     9m_2 = 18\qk - 4\eps\qi(2a+b),\\
     6k_1 = 6\qk + \eps\qi(a+2b),\quad
    18k_2 = 18\qk - \eps\qi(5a+16b).
  \end{gather*}
  The resulting linkage for $a = b = 1$ is depicted in
  \autoref{fig:pascal}, left and bottom. In this case, the curve is a
  cardioid. The joints are labeled by their corresponding rotation
  quaternions. The special geometry allows to realize the links
  connecting $h_1$, $h_2$, $m_1$ and $k_1$, $k_2$, $m_1$,
  respectively, by straight line segments. The linkage has two flat
  positions where all joints are collinear. There, either of the two
  anti-parallelograms ($h_i$, $m_i$, $k_i$, $m_{i-1}$ for
  $i \in \{1,2\}$) can switch to parallelogram mode. Thus, the motion
  has four components and only one of them is relevant for drawing the
  cardioid. Using $m_0 = 2\qk + 2\eps\qj$, we obtain a different, less
  symmetric, linkage to draw the same curve (\autoref{fig:pascal},
  right). Here, the triangles $h_1$, $m_1$, $h_2$ and $k_1$, $m_1$,
  $k_2$ act as links, that is, they are rigid throughout the
  motion. The motion of the link attached to $h_2$ and $m_2$ is the
  same for both linkages.
\end{example}

\begin{figure}
  \centering
  \begin{minipage}[t]{0.5\linewidth}
    \centering
    \rule{\linewidth}{0pt}
    \includegraphics{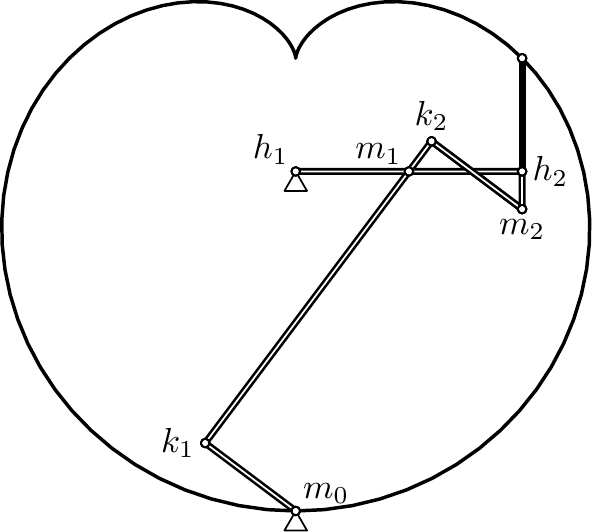}
  \end{minipage}
  \begin{minipage}[t]{0.5\linewidth}
    \centering
    \rule{\linewidth}{0pt}
    \includegraphics{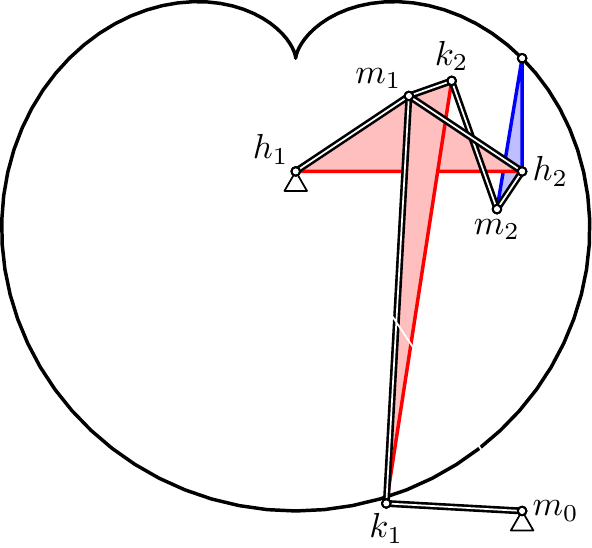}
  \end{minipage}
  \includegraphics{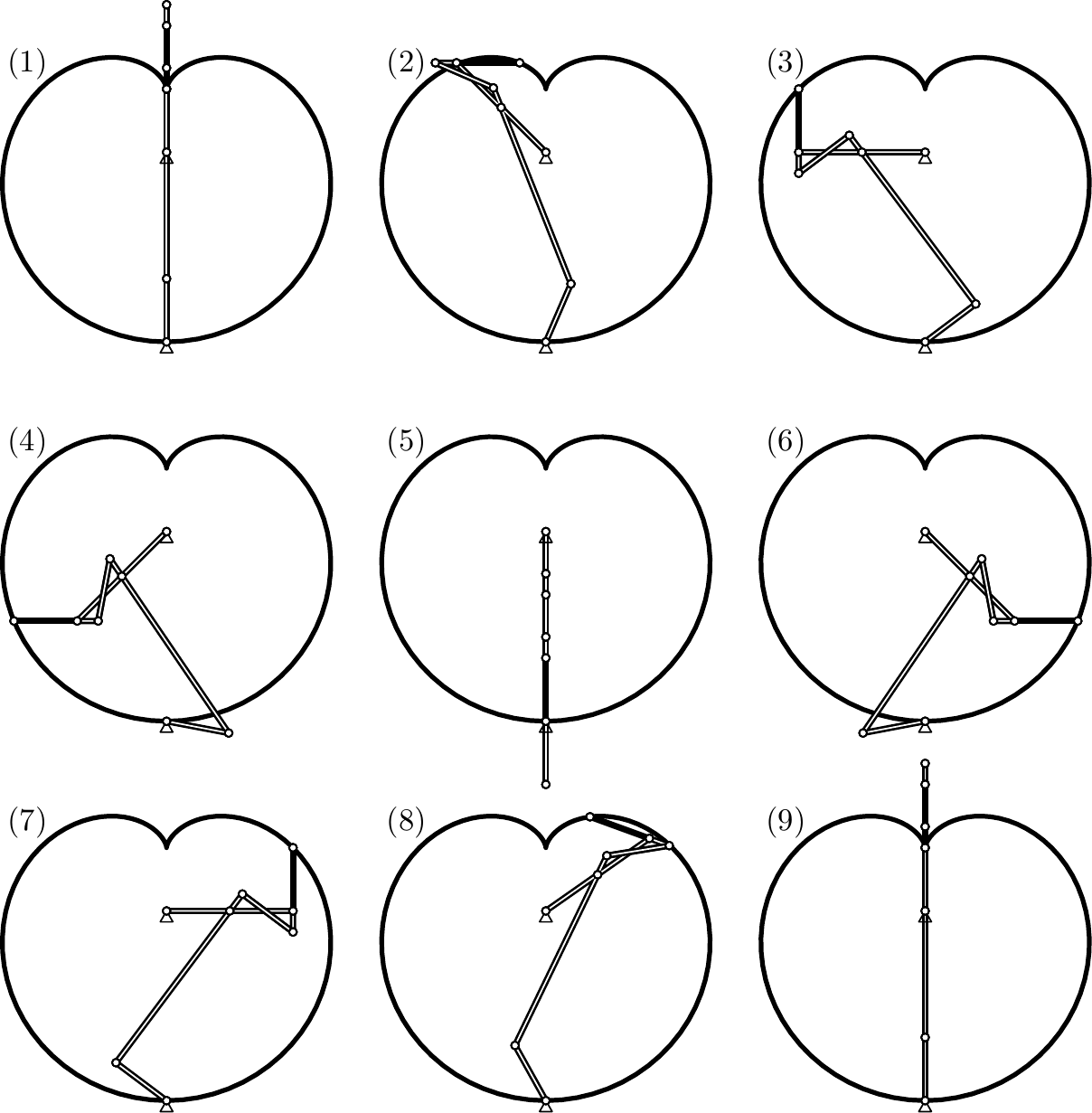}
  \caption{Linkages to draw a cardiod}
  \label{fig:pascal}
\end{figure}

\begin{example}
  We continue the discussion of Viviani's curve. However, we will not
  use \eqref{eq:3} but the simpler motion polynomial
  $C = (t-\qk)(t-\qj)$ which is obtained from \eqref{eq:3} by a
  translation. Based on this minimal motion we construct a linkage. We
  initialize the construction with $m_0 = \tfrac{1}{2}\qj$ and, using
  the recursion \eqref{eq:6}, compute
  \begin{equation*}
    10m_1 = -3\qj + 4\qk,\quad
    26m_2 = 5\qj - 12\qk,\quad
    5k_1 = 4\qj + 3\qk,\quad
    65k_2 = 33\qj+56\qk.
  \end{equation*}
  It is already apparent that this linkage has a flat folded position
  in the plane spanned by $\qj$ and $\qk$. One of its configurations
  is depicted in \autoref{fig:viviani}, left. We may as well construct a
  different linkage by starting with $m_0 = \tfrac{1}{2}\qi$ whence we get
  \begin{equation*}
    10m_1 = -3\qi + 4\qk,\
    50m_2 = 9\qi + 20\qj - 12\qk,\
    5k_1 = 4\qi + 3\qk,\
    25k_2 = -12\qi + 15\qj + 16\qk.
  \end{equation*}
  The corresponding linkage is shown in \autoref{fig:viviani},
  right. Here, the joint triples $(h_1, h_2, m_1)$ and
  $(k_1, k_2, m_1)$ are not collinear. Note that maybe more natural
  choices like $m_0 = \qi$ or $m_0 = \qj$ would violate the condition
  $\minpol(m_0) \neq \minpol(h_1)$.
\end{example}

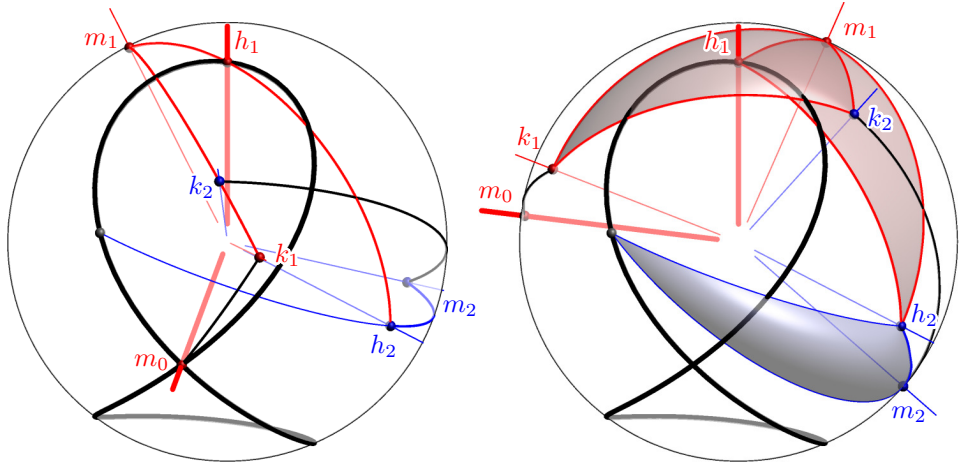
\begin{figure}
  \centering
  \begin{overpic}{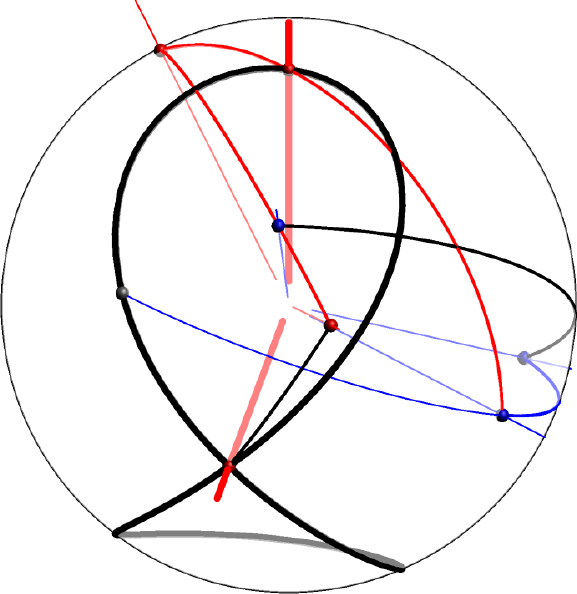}
    \contourlength{1pt}
    \put(50,91){\contour{white}{\textcolor{red}{$h_1$}}}
    \put(17,92){\textcolor{red}{$m_1$}}
    \put(59,44){\contour{white}{\textcolor{red}{$k_1$}}}
    \put(28,21){\contour{white}{\textcolor{red}{$m_0$}}}
    \put(80,24){\contour{white}{\textcolor{blue}{$h_2$}}}
    \put(40,59){\contour{white}{\textcolor{blue}{$k_2$}}}
    \put(96,33){\textcolor{blue}{$m_2$}}
  \end{overpic}
  \hfill
  \begin{overpic}{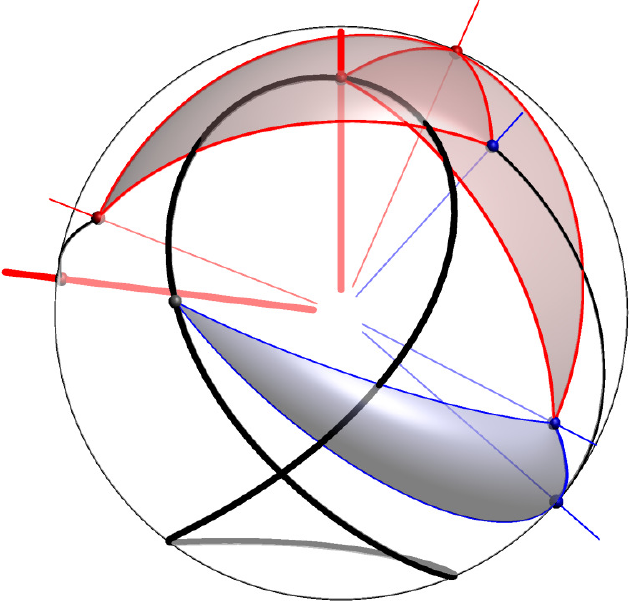}
    \contourlength{1pt}
    \put(47,86){\contour{white}{\textcolor{red}{$h_1$}}}
    \put(76,89){\textcolor{red}{$m_1$}}
    \put(8,66){\textcolor{red}{$k_1$}}
    \put(0,55){\textcolor{red}{$m_0$}}
    \put(90,29){\contour{white}{\textcolor{blue}{$h_2$}}}
    \put(81,70){\contour{white}{\textcolor{blue}{$k_2$}}}
    \put(86,9){\textcolor{blue}{$m_2$}}
  \end{overpic}
  \caption{Spherical linkages to draw Viviani's curve}
  \label{fig:viviani}
\end{figure}

\begin{example}
  In our next example, we return to the elliptic translation of
  \autoref{ex:5} and start with the planar factorization
  \eqref{eq:5}. We assume $a \neq 1$, choose $m_0 = -a\qk - b\eps\qj$
  and compute
  \begin{equation}
    \label{eq:7}
    \begin{gathered}
      (1+a)m_1 = -a(a+1)\qk - b\eps\qj(a - 1),\quad
      (1-a)m_2 = a(a-1)\qk - (a^2-2ab+b)\eps\qj,\\
      (1-a)^2m_3 = -a(a-1)^2\qk - (3a^2b-2a^2-b)\eps\qj,\quad
      (1+a)k_1 = -(a+1)\qk - 2b\eps\qj,\\
      2(1-a^2)k_2 = -2(a^2-1)\qk + (a^3-a^2(b-2)-a(6b-a)+3b)\eps\qj,\\
      2(1-a)^2k_3 = 2(a-1)^2\qk + (a^3+a^2(b-4)+a(8b-1)-5b)\eps\qj.
    \end{gathered}
  \end{equation}
  The centers of the rotation quaternions in \eqref{eq:7} describe the
  linkage in the configuration at $t = \infty$ and we can verify that
  all joints lie on the first coordinate axis. The linkage is quite
  similar to the linkage of \autoref{ex:pascal} but requires eight
  links and ten joints and it is similar to the example of
  \cite{gallet15} which was constructed in a similar manner. The
  linkage has two flat positions at which any of the three
  anti-parallelograms $m_{i-1}$, $k_i$, $m_i$, $h_i$ may switch to
  parallelogram mode. Thus, the configuration curve has six components
  and only one is relevant for drawing the ellipse.

  \begin{figure}
    \centering
    \includegraphics{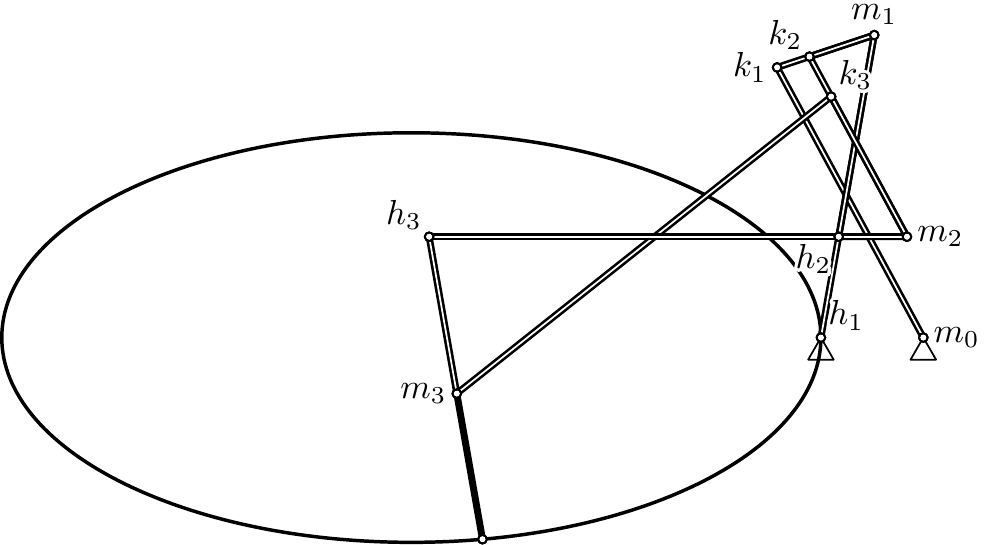}
    \includegraphics{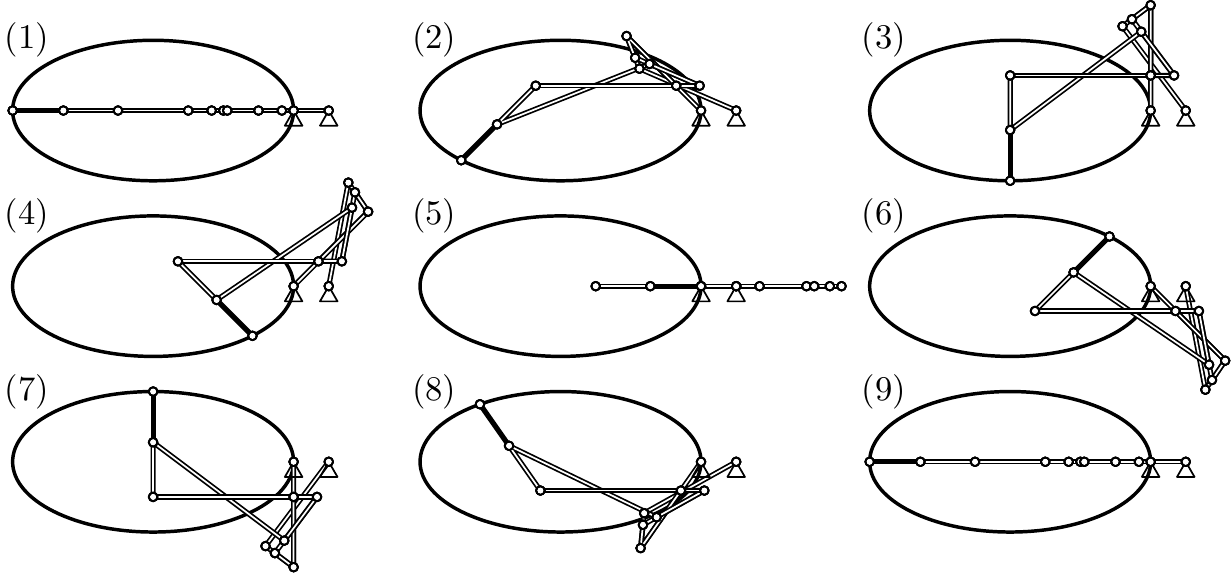}
    \caption{Linkage to draw an ellipse}
    \label{fig:elliptic-translation2}
  \end{figure}

  We may also construct a spatial linkage based on the factorization
  \eqref{eq:4} and the rotation quaternion $m_0 = 1 + \qj + \eps\qk$.
  We refrain from displaying the values of $m_1$, $m_2$, $m_3$, $k_1$,
  $k_2$, and $k_3$ (which would still be possible, even without
  specifying $a$ and $b$) and rather discuss the resulting spatial
  linkage (\autoref{fig:elliptic-translation}). It shares the
  linkgraph with the previous example and can be thought of as a
  scissor linkage made of Bennett linkages.
  \autoref{fig:elliptic-translation} displays a schematic
  representation where the four links and joints of one Bennett
  linkage are displayed in the same color (blue, gray, and red,
  respectively; note that the joints $m_1$ and $m_2$ belong to two
  Bennett linkages). The axis triples $(h_1,m_1,h_2)$,
  $(k_1,m_1,k_2)$, $(h_2,m_2,h_3)$, and $(k_2,m_2,k_3)$ are rigidly
  connected and one point attached to $h_3$ and $m_3$ draws the
  ellipse. In contrast to all other examples so far, the configuration
  curve of this linkage is irreducible. One can show that not only one
  but all trajectories are ellipses (or line segments) but in
  non-parallel plane. Thus, the linkage generates a so-called Darboux
  motion \cite{li15c}. It is possible to remove $m_1$ or $m_2$ without
  increasing the dimension of the configuration curve but this may
  come at the cost of introducing spurious components of the
  configuration curve.
\end{example}

\begin{figure}
  \centering
  \begin{overpic}{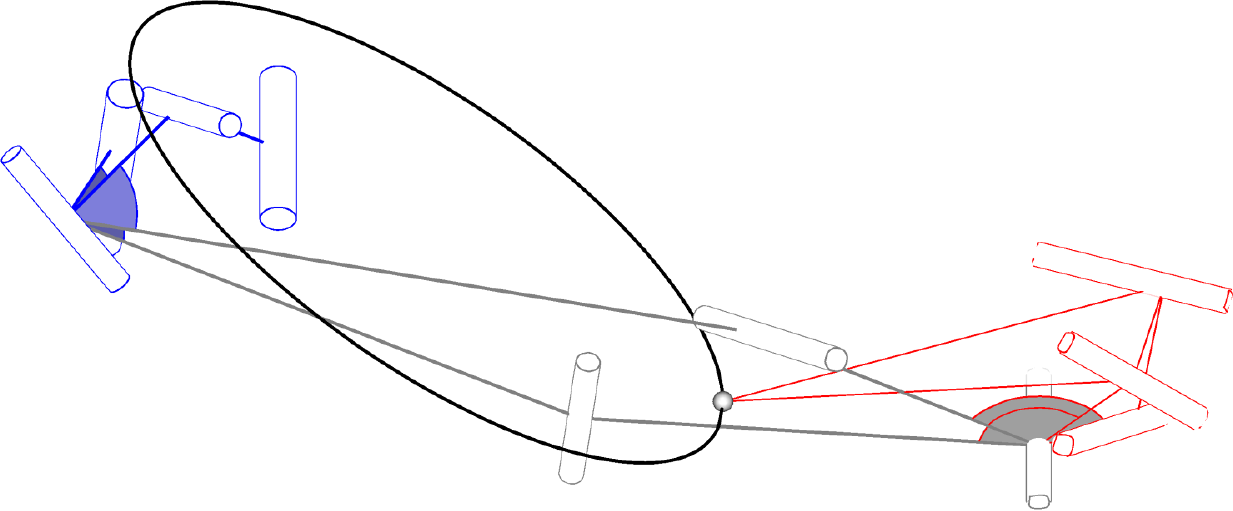}
    \put(5,33){\textcolor{blue}{$h_1$}}
    \put(25,33){\textcolor{blue}{$m_0$}}
    \put(15,34){\textcolor{blue}{$k_1$}}
    \put(1,22){\textcolor{blue}{$m_1$}}
    \put(49,10){\textcolor{black}{$h_2$}}
    \put(62,17){\textcolor{black}{$k_2$}}
    \put(79,1){\textcolor{black}{$m_2$}}
    \put(89,2){\textcolor{red}{$k_3$}}
    \put(95,10){\textcolor{red}{$h_3$}}
    \put(93,20){\textcolor{red}{$m_3$}}
  \end{overpic}
  \caption{Spatial linkage that can draw an ellipse}
  \label{fig:elliptic-translation}
\end{figure}
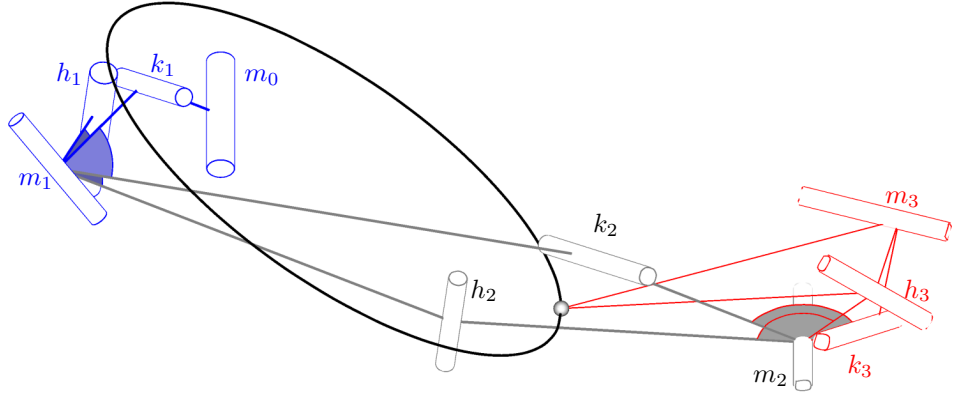

\begin{example}
  Finally, we illustrate how to draw a bounded portion of an unbounded
  rational curve. Using \eqref{eq:7} with $b = 0$ is prevented because
  then the choice $m_0 = -a\qk - b\eps\qj$ is invalid. But using the
  curve $x = t^2+1-2\qi$, the minimal motion $C = t^2 + 1 + \eps\qi$,
  the factor, $H = t - \qk$ and the rotation quaternion
  $m_0 = 2\qk + \frac{3}{4}\eps\qj$ produces the linkage shown in
  \autoref{fig:straight-line}.
\end{example}

\begin{figure}
  \centering
  \includegraphics[trim=5 70 60 0,clip,page=13]{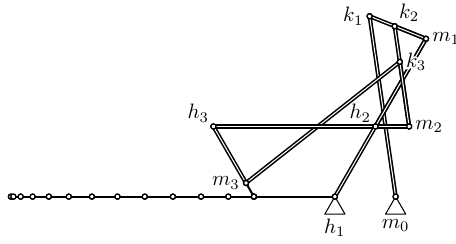}
  \caption{Linkage to draw a straight-line segment}
  \label{fig:straight-line}
\end{figure}

\section{Discussion of Results and Future Research}
\label{sec:discussion}

In order to carry out the computations for this article, we wrote an
experimental implementation in the computer algebra system Maple (version
18). Polynomial arithmetic, in particular the computation of $\gcd$-s, is
already available in Maple. Factorization over the real algebraic closure is not
available, but we worked around that by figuring out, in each example, which
field extension is needed. Maple can do the factorization in a given number
field; in most of our examples, especially in all examples occurring in
Section~\ref{sec:examples}, we do only examples with all factors defined
over~$\Q$.

For applications in engineering, it would be better to have an algorithm that
works with floating point numbers with a fixed precision. Exact
$\gcd$-computation is then not possible, but approximate versions of $\gcd$-s do
exist, for instance \cite{qrgcd04}. In order to adapt the algorithms in this
paper to the approximate setup, one needs to analyze the consequences of errors
carefully. We intend to do this in the future.

It is natural to ask whether our construction allows extensions to other joint
types. In principle, one could use prismatic (translation) joints to draw
unbounded rational curves. An extension of our algorithms to this case might be
possible, but we do not yet really know how (so this is another topic for future
research). The main problem is the failure of \autoref{th:3} for certain
unbounded polynomials. One such example is $t^2-\eps \qi.$ We could not find a
factorization into linear motion polynomials even after multiplication with a
quaternion polynomial.

By suitably selecting the dual quaternion $m_0$ in the construction of the
scissor linkage (\autoref{sec:scissor-linkage}) it is possible to create
intersecting revolute axes (``spherical joints'') in \emph{some} of the involved
four-bar linkages. In case of rational motions of degree two, we may produce a
linkage composed of two spherical four-bar linkages that has been called
``spherically constrained spatial revolute-revolute chain'' in the recent paper
\cite{abdul-sater16}. Our approach via Bennett flips is different from that
paper but may be used to improve certain aspects in the design process of the
car door guiding linkage that was presented there. This is actually the topic of
an ongoing research and demonstrates that ideas we presented in this paper may
be of engineering relevance.

\section*{Acknowledgments}
\label{sec:acknowledgements}

This work was supported by the Austrian Science Fund (FWF): P~26607
(Algebraic Methods in Kinematics: Motion Factorisation and Bond
Theory).

\bibliographystyle{amsplain}
\providecommand{\bysame}{\leavevmode\hbox to3em{\hrulefill}\thinspace}
\providecommand{\MR}{\relax\ifhmode\unskip\space\fi MR }
\providecommand{\MRhref}[2]{
  \href{http://www.ams.org/mathscinet-getitem?mr=#1}{#2}
}
\providecommand{\href}[2]{#2}

\end{document}